\documentclass[journal,onecolumn,12pt,twoside]{IEEEtranTCOM}

\usepackage[utf8]{inputenc} 
\usepackage[T1]{fontenc}
\usepackage{ifthen}
\usepackage{cite}
\usepackage[cmex10]{amsmath} 
\usepackage{mathrsfs}

\usepackage[end]{algpseudocode}
\usepackage{amssymb}
\usepackage{mathtools}
\usepackage{bm}
\usepackage{subfig}
\usepackage{color}
\usepackage{amsthm}

\usepackage[pdftex]{graphicx}
\usepackage{amsmath}
\usepackage{array}
\usepackage{url}
\usepackage{multicol}
\usepackage{multirow}
\usepackage{algorithm}
\usepackage{booktabs}
\usepackage{diagbox} 
\usepackage{setspace}

\newcommand{\mtx}[1]{\bm{#1}}

\DeclarePairedDelimiter{\ceil}\lceil\rceil

\newtheorem{theorem}{Theorem}
\newtheorem{lemma}{Lemma}
\newtheorem{proposition}{Proposition}
\newtheorem{remark}{Remark}

\newtheorem{corollary}{Corollary}

\ifCLASSINFOpdf
\else
\fi
\begin{document}
%
\title{Optimal Load Allocation for Coded Distributed Computation in Heterogeneous Clusters}
%
%
%

\author{DaeJin~Kim,  Hyegyeong~Park, and~Junkyun~Choi
\thanks{D. Kim and J. Choi are with the School of Electrical Engineering, Korea Advanced Institute of Science and Technology (KAIST), Daejeon, 34141, South Korea (e-mail: deejay@kaist.ac.kr; jkchoi59@kaist.edu).}
\thanks{H. Park is with the Computer Science Department, Carnegie Mellon University, Pittsburgh, PA 15213 USA (e-mail: hyegyeop@cs.cmu.edu).}
}

\maketitle

\begin{abstract}
Recently, coding has been a useful technique to mitigate the effect of \textit{stragglers} in distributed computing. However, coding in this context has been mainly explored under the assumption of \textit{homogeneous} workers, although the real-world computing clusters can be often composed of \textit{heterogeneous} workers that have different computing capabilities. The uniform load allocation without the awareness of heterogeneity possibly causes a significant loss in latency. In this paper, we suggest the optimal load allocation for coded distributed computing with heterogeneous workers.
Specifically, we focus on the scenario that there exist workers having the same computing capability, which can be regarded as a group for analysis.
We rely on the lower bound on the expected latency and obtain the optimal load allocation by showing that our proposed load allocation achieves  the minimum of the lower bound for a sufficiently large number of workers. From numerical simulations, when assuming the group heterogeneity, our load allocation reduces the expected latency by orders of magnitude over the existing load allocation scheme.
\end{abstract}

%


%
\IEEEpeerreviewmaketitle

\section{Introduction}
%
%
%
%

\IEEEPARstart{D}{istributed} computing has become the mainstream for most computing platforms \cite{Dean12Large, Dean08Mapreduce, Zaharia10Spark} to support today's increasing demand for handling large-scale data and computational workloads. Splitting the computation task into the multiple sub-computations and allocating them to multiple computing nodes, distributed computing improves the system latency enjoying the virtue of concurrent/parallel processing of subtasks. Moreover, the inherent scalability, which means that users can easily add workers to the system as the demand for computing grows, leads distributed computing to
the emerging de facto standard for many modern computing architectures.

However, some straggling workers 
from various sources of delay including transient and permanent failures can potentially limit the performance of distributed computing systems, 
since commodity computing nodes are often deployed to enhance the cost-effectivity. 
To address this issue, in \cite{Lee18Speeding} the authors introduce the notion of \textit{coded computation} which exploits
the redundant sub-computations to speed up the distributed matrix-vector multiplication.
More specifically, coding improves the computation latency by allowing a subset of the computation results to complete the overall computation.  For example, to multiply a matrix $\boldsymbol{A} (\in \mathbb{R}^{2k\times d})$ with a vector $\mathbf{x} (\in \mathbb{R}^{d\times 1})$, one can split $\boldsymbol{A}$ into two submatrices, i.e., $\boldsymbol{A} = [\boldsymbol{A}_1; \boldsymbol{A}_2]$ where $\boldsymbol{A}_1 \in \mathbb{R}^{k\times d}$ and $\boldsymbol{A}_2 \in \mathbb{R}^{k\times d}$. The submatrices are encoded to $[\boldsymbol{A}_1; \boldsymbol{A}_2; \boldsymbol{A}_1+\boldsymbol{A}_2]$ and the three subtasks $\boldsymbol{A}_1 \mathbf{x}$, $\boldsymbol{A}_2 \mathbf{x}$ and  $(\boldsymbol{A}_1 +\boldsymbol{A}_2) \mathbf{x}$ are assigned to the three distributed workers, respectively. Then, collecting any two out of three subtask results would be sufficient for obtaining $\boldsymbol{A}\mathbf{x}$, which is directly translated into the tolerance to the one-straggling worker and reduction in computation latency.

Galvanized by the work of \cite{Lee18Speeding}, the idea of coded computation using \textit{homogeneous} workers has been widely used for the various types of computations: coding for speeding up high-dimensional matrix-matrix multiplications \cite{Lee17High-dimensional, Yu17Polynomial, Baharav18Straggler-Proofing, Dutta18Optimal, Park19Irregular}, matrix multiplication with the awareness of practical computing cluster architectures \cite{Park18Hierarchical, Li17Coding, Gupta17Locality}, distributed optimization \cite{Karakus17Straggler, Karakus17Encoded, Ferdinand19Anytime, Zhu17A}, gradient descent \cite{Tandon17Gradient, Raviv17Gradient, Halbawi17Improving, Ozfaturay18Speeding,Ozfatura18Distributed,Yue20Coded}, convolution \cite{Dutta17Coded}, distributed inference and transmission in mobile edge computing \cite{Zhang19OnModel}, and node-selection based subtask assignment method \cite{Zhao19ANode}.
 The work of \cite{Ferdinand18Hierarchical, Kiani18Exploitation, Mallick18Rateless} deals with exploiting stragglers' computation results. In \cite{Yapar19Device}, the authors propose a coded caching scheme for a device-to-device network with straggling servers.
A learning-based code design to approximate unavailable computation results is suggested in \cite{Kosaian18Learning}.

Practical large-scale systems, however, usually consist of \textit{heterogeneous} workers with different computing capabilities \cite{Zaharia08Improving}.
The uniform allocation of the workload to workers with the ignorance of heterogeneity potentially leads to the loss of system performance. This naturally induces our main problem: how can we optimally allocate the workload to the heterogeneous workers in the distributed matrix-vector multiplication?
Specifically, we focus on the scenario assuming the \textit{group heterogeneity}, which means that the workers having the same computing capability are regarded as a \textit{group} for analysis. 
This group heterogeneity comes from, for example, the incremental deployment of groups of computing machines in data centers.
 Note that our modeling is not limited to the computing cluster architectures with physically grouped workers.

Despite its importance, to the best of our knowledge, only the work of \cite{Reisizadeh19Coded} and  \cite{Kim19Coded} considers the latency of coded matrix multiplication under the assumption of heterogeneous workers.
The work of \cite{Reisizadeh19Coded} resorts to the two-step alternative problem formulation maximizing the expected number of aggregated subtask results from the workers.
We take a different approach to obtain the optimal load allocation; 
we rely on the lower bound on the expected latency and obtain the optimal load allocation by showing that our proposed load allocation achieves  the minimum of the lower bound for a sufficiently large number of workers. 
The group heterogeneity reflecting the constraints in the practical computing clusters is a key assumption that enables this lower bound approach to yield the optimal load allocation. 
Relative to \cite{Reisizadeh19Coded}, our work suggests the theoretical lower bound for the expected latency and provides a simpler way of proofs thanks to the group heterogeneity. 
Moreover, the proposed analysis is applicable to the regime that the complexity of problem does not depend on the size of computing system, while the analysis of \cite{Reisizadeh19Coded} requires a linearly scaling problem complexity in the number of workers.
In modeling of \cite{Kim19Coded}, every worker is assigned the same fixed number of rows of the uncoded data matrix $\boldsymbol{A}$ even with an increasing number of the workers, which leaves a room for improvement.
By relaxing the above condition in \cite{Kim19Coded}, our load allocation reduces the latency by orders of magnitude over the allocation in \cite{Kim19Coded} as the number of workers increases.


\begin{figure}[t]
	\centering{
\vspace{-0.15in}
 \includegraphics[width=1\linewidth]{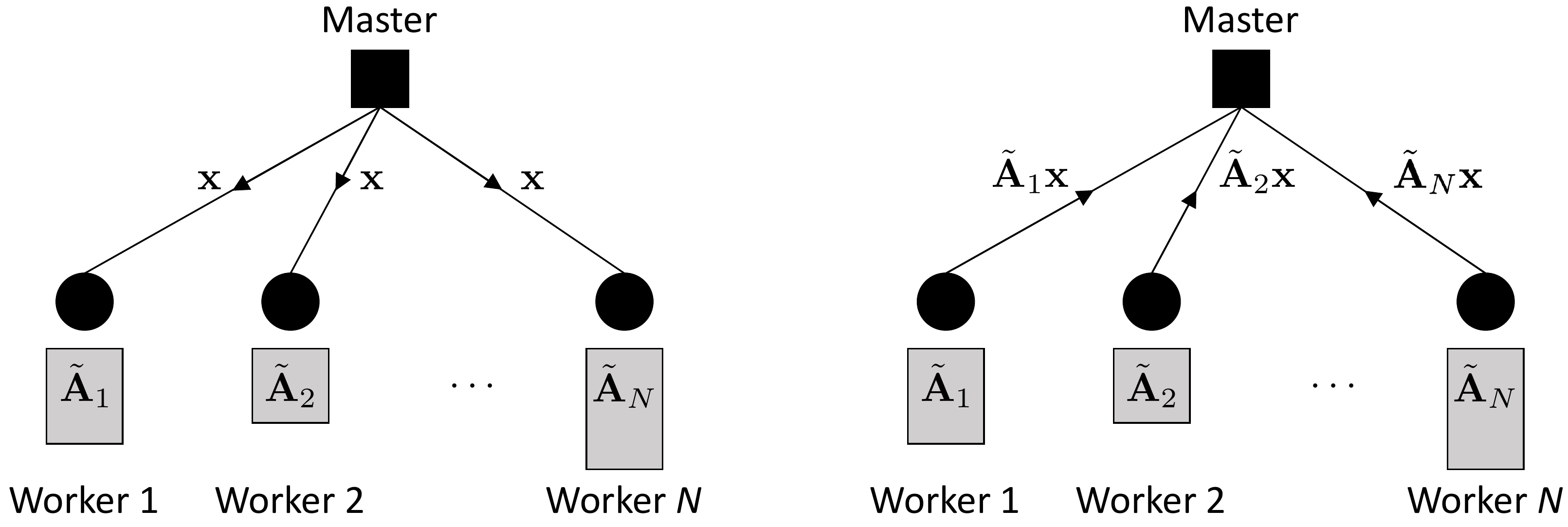} }
\caption{The master sends the input vector $\mathbf{x}$ to the $N$ workers, each of which stores $\boldsymbol{\tilde{A}}_i$. Worker $i$ computes $\boldsymbol{\tilde{A}}_i \mathbf{x}$ with straggling parameter $\mu_i$, and sends back the computation result to the master.}
	\label{Fig:Modeling}
\vskip -0.2in
\vspace{-0.1in}
\end{figure}

\textit{Contribution:} The key contributions of this work are summarized as follows.
\begin{itemize}
	\item We provide a new proof method to derive the optimal load allocation for computing clusters that consist of heterogeneous workers under the classical probabilistic model for latency.
	\item For a sufficiently large number of workers, we present the optimal load allocation for heterogeneous workers by finding a lower bound of the expected latency and its achievable scheme. 
	\item We demonstrate the optimal design of the $(n,k)$ maximum distance separable (MDS) code to achieve the minimum expected latency based on the proposed load allocation.

\end{itemize}

\textit{Notations: }
We denote $[n]=\{1,2, \ldots, n\}$ for $n \in \mathbb{N}$.  Let $\mathscr{P}([n])$ be the power set of $[n]$. For any $u \in [n]\cup \{ 0 \}$, we define $N^u = \{ A\in \mathscr{P}([n]) : |A| = u \}$. For nonnegative functions $f(n)$ and $g(n)$, we denote $f(n) = \Theta(g(n))$ if there exist $k_1 > 0$, $k_2 > 0$, $n_0 \in \mathbb{N}$ such that if $n>n_0$, then $k_1g(n) \le f(n) \le k_2 g(n)$. The ceil function of $x (\in \mathbb{R})$, denoted by $\ceil{x}$, returns the smallest integer that is greater than or equal to $x$.

\section{System Model, Model Assumptions, and Problem Formulation} \label{Sec:SystemModel}

We focus on the distributed matrix-vector multiplication over a master-worker setup in heterogeneous clusters. In this section, we describe our computation model and runtime distribution model.

\subsection{Computation Model} 
We assume that there are $N$ workers that are divided into $G$ groups, each of which has a different number of workers and different runtime distribution.\footnote{Although our modeling assumes the group heterogeneity, the latency analysis can be extended to approximate the latency of the computing system with ``fully'' heterogeneous workers by grouping the workers based on the reasonable off-the-shelf clustering methods.}
In this modeling, group $j$ consists of $N_j$ workers, i.e., $\sum_{j \in [G]} N_j = N$. We assign a computation task to multiply a given data matrix $\mtx{A} \in\mathbb{R}^{k \times d}$ with an input vector $\mathbf{x}\in \mathbb{R}^{d \times 1}$ to the $N$ distributed workers.
We apply an $(n, k)$ MDS code to the rows of $\mtx{A}$ for obtaining the coded data matrix $\tilde{\mtx{A}}\in \mathbb{R}^{n \times d}$. Afterwards, the rows of $\tilde{\mtx{A}}$ are grouped into $N$ submatrices as $\tilde{\mtx{A}} = [\tilde{\mtx{A}}_1; \tilde{\mtx{A}}_2; \dots ; \tilde{\mtx{A}}_N]$,
where $\tilde{\mtx{A}}_i \in \mathbb{R}^{l_i \times d}$ is the coded data matrix allocated to worker $i$ and $n = \sum^N_{i=1} l_i  .$
 It is assumed that workers in each group $j$ are assigned the coded data matrix with the same number of rows, denoted by $l_{(j)}$. 
Then, worker $i$ is assigned a subtask to compute $\tilde{\mtx{A}}_i\mathbf{x}$ and  sends back the product to the master upon finishing its subtask as shown in Fig. \ref{Fig:Modeling}. The master can retrieve $\mtx{A}\mathbf{x}$ by collecting the inner product of any $k$ coded rows with $\mathbf{x}$ due to the MDS property.


\subsection{Runtime Distribution Model}
Let $T^{(j)}$ denote the random variable representing the round trip time taken for calculating the inner product of $l_{(j)}$ rows of $\tilde{\mtx{A}}$ with $\mathbf{x}$ at a worker in group $j$.
We assume that the random variable $T^{(j)}$
follows a shifted exponential distribution with rate $\mu_{(j)}$ as follows:
\begin{equation} \label{Eq:CDF_in_Group_j}
F_{j}(t) = \Pr(T^{(j)}\le t) = 1-e^{-\frac{k \mu_{(j)}}{l_{(j)}} \left (t-\frac{\alpha_{(j)} l_{(j)}}{k} \right)}
\end{equation}
for $t \ge \frac{\alpha_{(j)} l_{(j)}}{k}$ and $j\in [G]$, where $\alpha_{(j)}$ is the shift parameter of a worker in group $j$.
The probabilistic model is motivated by the model proposed in \cite{LiangTOFEC14} which is used for modeling latency of file queries from cloud storage systems. The distribution has been widely accepted in the existing literature \cite{Lee18Speeding,Reisizadeh19Coded,Dutta17Coded}.
As demonstrated in \cite{Lee18Speeding,Reisizadeh19Coded}, the shifted exponential model provides a good fit for the runtime distribution over cloud computing platform such as AWS EC2. Moreover, the shifted exponential distribution provides an adequate balance between accuracy and analytical tractability.


Considering the homogeneous master-worker model as in \cite{Lee18Speeding}, i.e., $G = 1$ and $\alpha_{(j)} = 1$, the number of rows allocated to each worker becomes $l_{(j)} = \frac{k}{N}$.
It follows that the probabilistic model in \eqref{Eq:CDF_in_Group_j} is equal to that of \cite{Lee18Speeding}.
In addition, the difference between the model represented in \eqref{Eq:CDF_in_Group_j} and the model in \cite{Reisizadeh19Coded} is as follows:
This paper assume that the cumulative distribution function (CDF) of task runtime for a worker with straggling parameter $\mu_{(j)}$ and shift parameter $\alpha_{(j)}$ to calculate $k$ rows is
\begin{equation} \label{Eq:CDFmother}
1-e^{-\mu_{(j)}(t-\alpha_{(j)})}.
\end{equation}
On the other hand, the paper in \cite{Reisizadeh19Coded} assume that  the CDF of task runtime for a worker with straggling parameter $\mu_{(j)}$ and shift parameter $\alpha_{(j)}$ to calculate one row is equal to  \eqref{Eq:CDFmother}.


%
%
%
%

\subsection{Problem Formulation}
 For given an input data matrix with $k$ rows and $G$ groups of workers with the straggling parameter $\mu_{(j)}$ for workers in group $j$, 
 we are interested in obtaining the optimal load allocation $(l^\ast_{(1)}, l^\ast_{(2)}, \dots, l^\ast_{(G)})$ and designing the $(n,k)$ MDS codes to minimize the expected computation latency. 
 Due to the heterogeneity of the straggling parameters of the workers, the expected computation latency cannot be directly calculated from the known result of order statistics. We thus take a detour as described in this subsection and show the asymptotic optimality of our solution in Section \ref{Sec:Optimal_load_allocation}.
 
 Let $r_j$ denote the number of workers in group $j$ which finish the assigned subtasks when the master completes to receive the inner product of $k$ rows of $\tilde{\mtx{A}}$ with $\mathbf{x}$. Then, we denote the summation of $r_j$'s for all groups by $r$, i.e., $r = \sum_{j\in [G]} r_j$. 
For $N = \sum_{j\in [G]} N_j$, let $T_{r:N}$ denote the $r$-th order statistic of $N$ exponential random variables following the distribution given in \eqref{Eq:CDF_in_Group_j} for each of $N_j$ workers that belongs to group $j$.
We aim at finding the optimal load allocation to minimize the expected computation time $E[T_{r:N}]$ for all $r \in [N]$. Throughout the remainder of our paper, $E[T_{r:N}]$ is denoted by $\lambda_{r:N}$ for notational convenience.

In our modeling, group $j~\in[G]$ consists of $N_j$ workers with straggling parameter $\mu_{(j)}$ and shift parameter $\alpha_{(j)}$. 
Workers in group $j$ are assigned the coded data matrix with $l_{(j)}$ rows. We thus have 
\begin{equation}\label{Eq:Constraints_n}
n = \sum_{j\in [G]} N_j l_{(j)}  .
\end{equation}

In addition, we assume that $N_j = \Theta(N)$ for $j\in [G]$. Recall that an $(n,k)$ MDS code is applied to the rows of $\mtx{A}\in\mathbb{R}^{k \times d}$ for obtaining $\tilde{\mtx{A}} \in \mathbb{R}^{n \times d}$. This implies that the master needs to collect the inner product of the $k$ rows of $\tilde{\mtx{A}}$ with $\mathbf{x}$ to retrieve $\mtx{A}\mathbf{x}$. In this sense, the condition for guaranteeing the successful recovery of $\mtx{A}\mathbf{x}$ at the master is given as
\begin{equation} \label{Eq:Constraint}
\sum_{i \in B} l_i = k \hspace{0.5in} \textnormal{ for } B\in N^r . 
\end{equation}
Then, \eqref{Eq:Constraint} is rewritten as 
\begin{equation} \label{Eq:Sumk}
\sum_{j \in [G]} r_j l_{(j)} = k . 
\end{equation}


For given $\left(\mu_{(1)}, \mu_{(2)}, \dots, \mu_{(G)} \right)$, $\left(\alpha_{(1)}, \alpha_{(2)}, \dots, \alpha_{(G)} \right)$, $k$, and $\left(N_1, N_2, \dots, N_G\right)$, our main objective is finding the optimal $(n^*,k)$ MDS code and load allocation $\left( l_{(1)}^*, l_{(2)}^*,\dots, l_{(N)}^* \right)$ to have the minimum $\lambda_{r:N}$ for all $r \in [N]$ with the constraint (\ref{Eq:Sumk}). Note that the optimal load allocation $\left( l_{(1)}^*, l_{(2)}^*,\dots, l_{(N)}^* \right)$ leads us to choose the optimal $n^\ast$  from \eqref{Eq:Constraints_n}, which eventually means the optimal design of the $(n, k)$ MDS code.


\section{Optimal Load Allocation} \label{Sec:Optimal_load_allocation}

In this section, we provide the optimal load allocation method under the proposed system model. First, we find a condition for the optimal load allocation using a lower bound of $\lambda_{r:N}$. Next, we introduce the optimal load allocation method which achieves the minimum of the lower bound. Finally, we prove that for given the optimal load allocation, $\lambda_{r:N}$ is equal to the lower bound for sufficiently large $N$.

\subsection{Condition for Optimal Load Allocation} \label{Subsec:Add_condition}
Recall that $r_j$ denotes the number of workers in group $j$ which complete the assigned subtasks when the master receives the inner product of $k$ coded rows with $\mathbf{x}$ ($r = \sum_{j\in [G]} r_j$).
Let $T_{r_j:N_j}^{l_{(j)}}$ denote the $r_j$-th order statistic of $N_j$ random variables following the distribution in \eqref{Eq:CDF_in_Group_j}.
The expectation $E \left [T_{r_j:N_j}^{l_{(j)}} \right ]$ is rewritten as $\lambda_{r_j:N_j}^{l_{(j)}}$ for notational convenience.
Here, $\lambda_{r_j:N_j}^{l_{(j)}}$ is equivalent to the 
average runtime of a system using an $(N_j, r_j )$ MDS code.\footnote{Note that this does not mean that the group-wise small MDS codes are actually deployed.} 
At group $j$ we thus have
\begin{equation} \label{Eq:Group_expectation}
\lambda_{r_j:N_j}^{l_{(j)}} =  \frac{l_{(j)}}{k} \left( \alpha_{(j)} + \frac{1}{\mu_{(j)}} \log \left(\frac{N_j}{N_j - r_j} \right) \right) ,
\end{equation}
 for given $l_{(j)}$.
Here, we use an approximation of $\mathcal{H}_n - \mathcal{H}_{n-k} \approx \log (\frac{n}{n-k})$ where $\mathcal{H}_n = \sum_{i\in [n]} \frac{1}{i}$. We set aside the derivation of equation \eqref{Eq:Group_expectation}  in Appendix \ref{App:Group_Expectation}.

Recall that the overall latency is denoted by $T_{r:N}$.
Our goal is to obtain the optimal load allocation method to minimize $\lambda_{r:N}$. However, it is not easy to get a closed form for $\lambda_{r:N}$ for given $\boldsymbol{N} = ( N_1, N_2 \dots, N_G )$ and $\boldsymbol{\mu} = (\mu_{(1)}, \mu_{(2)}, \dots, \mu_{(G)})$.
Thus, we use the following approach to solve the problem.
Note that 
\begin{equation} \label{Eq:Order_max}
T_{r:N} = \max_{j\in [G]} \left\{ T_{r_j: N_j}^{l_{(j)}}\right\} .
\end{equation}
It follows from the definition of maximum that
$$
T_{r:N} = \max_{j\in [G]} \left\{ T_{r_j: N_j}^{l_{(j)}}\right\} \iff T_{r:N} \ge T_{r_j: N_j}^{l_{(j)}} \textnormal{ for all } j \in [G]  .
$$
Then, we have the lower bound for the expectation of $T_{r:N}$ as follows:
\begin{equation*}
\lambda_{r:N} = E \left[ \max_{j\in [G]} \left\{ T_{r_j: N_j}^{l_{(j)}}\right\} \right] \ge  \max_{j\in [G]} \left\{ \lambda_{r_j: N_j}^{l_{(j)}} \right\} . 
\end{equation*}

We first find the optimal load allocation to achieve the minimum of $\max_{j\in [G]} \left\{ \lambda_{r_j: N_j}^{l_{(j)}} \right\}$. Then, given the optimal load allocation, it will be shown in Section \ref{Subsec:Analysis_lambda} that $\lambda_{r:N}$ converges to $\max_{j\in [G]} \left\{ \lambda_{r_j: N_j}^{l_{(j)}} \right\}$ as $N$ goes to infinity.
From now on, $l_{(j)}$ and $r_j$ are considered as real values to make the analysis easy for $j\in [G]$. Observe the following lemma.
\begin{lemma}\label{Lmm:Lower_conditionG=2}
Consider that there are two groups, i.e., $G = 2$. Let $\boldsymbol{l}^* = (l_{(1)}^*,l_{(2)}^*)$ be the optimal load allocation which achieves the minimum of $\max_{j\in \{1, 2\}} \left\{ \lambda_{r_j: N_j}^{l_{(j)}} \right\}$.
Let  $\boldsymbol{r}^* = (r_1^*, r_2^*)$ be the numbers of workers to complete the assigned task corresponding to $(l_{(1)}^*,l_{(2)}^*)$.
Then, $\lambda_{r_1^*: N_1}^{l_{(1)}^*} = \lambda_{r_2^*: N_2}^{l_{(2)}^*}.$
\end{lemma}

\begin{proof}
Suppose $\lambda_{r_1^*: N_1}^{l_{(1)}^*} \neq \lambda_{r_2^*: N_2}^{l_{(2)}^*}$. Without loss of generality, we may assume $\epsilon := \lambda_{r_1^*: N_1}^{l_{(1)}^*} - \lambda_{r_2^*: N_2}^{l_{(2)}^*}  >0$. It follows from the MDS property that the pair $(\boldsymbol{l}^*, \boldsymbol{r}^*)$ satisfies the constraint $r_1^* l_{(1)}^* + r_2^* l_{(2)}^* = k$ for the successful recovery of $\mtx{A} \mathbf{x}$.
For fixed $\boldsymbol{r}^*$, let $\bar{l}_{(2)}$ satisfy the following equation:
\begin{equation} \label{Eq:Group2_balancing}
\lambda_{r_2^*: N_2}^{\bar{l}_{(2)}} = \lambda_{r_2^*: N_2}^{l_{(2)}^*} + \frac{\epsilon}{2} .
\end{equation}
We denote $\bar{l}_{(1)} = \frac{k - r_2^* \bar{l}_{(2)}}{r_1^*}$.
Since $\bar{l}_{(2)} > l_{(2)}^*$, we have $ l_{(1)}^* > \bar{l}_{(1)}$.
This implies that $\lambda_{r_1^*: N_1}^{{l}_{(1)}^*} > \lambda_{r_1^*: N_1}^{\bar{l}_{(1)}}  .$
 Then $\bar{l}_{(2)} > l_{(2)}^*$ and $\bar{l}_{(1)} < l_{(1)}^*$. 
It follows from (\ref{Eq:Group2_balancing}) that $\lambda_{r_1^*: N_1}^{l_{(1)}^*} > \lambda_{r_2^*: N_2}^{\bar{l}_{(2)}}  .$
Thus, we have $\max_{j\in \{1, 2\}} \left\{ \lambda_{r_j^*: N_j}^{l_{(j)}^*} \right\} > \max_{j\in \{1, 2\}} \left\{ \lambda_{r_j^*: N_j}^{\bar{l}_{(j)}} \right\} ,$
which is a contradiction to the the assumption that $(\boldsymbol{l}^*, \boldsymbol{r}^*)$ achieves the minimum of $\max_{j\in \{1, 2\}} \left\{ \lambda_{r_j: N_j}^{l_{(j)}} \right\}.$
\end{proof}

 Lemma \ref{Lmm:Lower_conditionG=2} means that the expected latency can be minimized through balancing the task by shifting the workload from a group with more workload to a group with less workload depending on the straggling parameter and the number of workers in each group.  
Next, we introduce Theorem \ref{Thm:Lower_condition}, a generalization of Lemma \ref{Lmm:Lower_conditionG=2}, indicating the minimality condition for $ \max_{j\in [G]} \left\{ \lambda_{r_j: N_j}^{l_{(j)}} \right\}$ $(G \ge 2 )$ given the optimal load allocation. Applying the same argument used in Lemma \ref{Lmm:Lower_conditionG=2}, Theorem \ref{Thm:Lower_condition} is verified.

\begin{theorem}\label{Thm:Lower_condition}
For given $G \ge 2$, let $\boldsymbol{l}^* = (l_{(1)}^*, l_{(2)}^*, \dots, l_{(G)}^* )$ be the optimal load allocation which achieves the minimum of $ \max_{j\in [G]} \left\{ \lambda_{r_j: N_j}^{l_{(j)}} \right\}$. Let $\boldsymbol{r}^* = (r_1^*, r_2^*, \dots, r_G^* )$ be the numbers of workers to complete the assigned task corresponding to $\boldsymbol{l}^*$. Then $\lambda_{r_j^*: N_j}^{l_{(j)}^*} = \lambda_{r_{j'}^*: N_{j'}}^{l_{(j')}^*}  ,$
 for all $j \neq j' \in [G]$.
\end{theorem}

\begin{proof}
The proof is in Appendix \ref{APP:Theorem1}.
\end{proof}

We define
\begin{equation} \label{Eq:xi}
\xi \left(r_j,N_j,\mu_{(j)} \right) =  \alpha_{(j)} + \frac{1}{\mu_{(j)}} \log \left(\frac{N_{j}}{N_{j} - r_{j}} \right) 
\end{equation}
for notational convenience.
From Theorem \ref{Thm:Lower_condition}, the optimal load allocation $(\boldsymbol{l}^*, \boldsymbol{r}^*)$ satisfies the following equations:
\begin{equation} \label{Eq:Mean_equality_condition}
l_{(j)} \xi \left( r_j, N_j, \mu_{(j)} \right)  = l_{(j')} \xi \left( r_{j'}, N_{j'}, \mu_{(j')} \right)   , \hspace{0.5in} \textnormal{ for } j \neq j' \in [G]
\end{equation}


\subsection{Determining Optimal Load Allocation and $(n,k)$ MDS Code for Achieving Lower Bound of $\lambda_{r:N}$}  \label{Subsec:Loadallocation}
We assume that $k$ is given. Note that we have the constraints (\ref{Eq:Sumk}) and (\ref{Eq:Mean_equality_condition}). 
In this subsection, we provide the optimal load allocation ($\boldsymbol{l}^*, \boldsymbol{r}^* $)  which achieves the minimum of  $\max_{j\in [G]} \left\{ \lambda_{r_j: N_j}^{l_{(j)}} \right\}$. 
In addition, for given $k$, we determine the $(n,k)$ MDS code to achieve the minimum of  $$\max_{j\in [G]} \left\{ \lambda_{r_j: N_j}^{l_{(j)}} \right\}.$$ 
We define a function 
\begin{equation}\label{Eq:Objective_G}
f ( \boldsymbol{r}) =  \frac{1}{ \sum_{j\in [G]} \frac{r_j}{\xi \left(r_j, N_j, \mu_{(j)} \right) }}  ,
\end{equation}
where $\boldsymbol{r} = (r_1, r_2, \dots, r_G) $.
\begin{lemma} \label{Lmm:Convexity_f(r)}
The function $f ( \boldsymbol{r} )$ in (\ref{Eq:Objective_G}) is a strictly convex function on an open set $S$, where $S$ is a Cartersian product of open intervals $(0, N_j)$, for $j \in [G]$, i.e., 
\begin{equation} \label{Eq:CartesianS}
S = \prod_{j\in [G]} (0, N_j)  .
\end{equation}
\end{lemma}

\begin{proof}
Clearly, $S$ is a convex set.
Let $g ( \boldsymbol{r}) = \sum_{j\in [G]} \frac{r_j}{\xi \left(r_j, N_j, \mu_{(j)} \right) }  .$
Then, it suffices to show that $ g ( \boldsymbol{r} ) $ is a strictly concave and positive function.
Since $\frac{N_j}{N_j - r_j} > 1$ for $j \in [G]$, it is clear that $ g (\boldsymbol{r} ) $ is positive on $S$. 
Let $g ( \boldsymbol{r} ) = \sum_{j \in [G]} g_j(\boldsymbol{r} )$, 
where  $$g_j(\boldsymbol{r} ) = \frac{r_j}{\alpha_{(j)} + \frac{1}{\mu_{(j)}} \log \left(\frac{N_j}{N_j - r_j} \right)}  .$$
One can show that $g_j(\boldsymbol{r})$ is concave. We set aside the proof of the concavity of $g_j(\boldsymbol{r})$ in Appendix \ref{App:g_j}.
Since the sum of concave functions is a concave function, we have that $g( \boldsymbol{r})$ is a strictly concave function on $S$ as desired.
\end{proof}


\begin{lemma} \label{Lmm:r_j^*proof}
Let $r_j^*$ be the solution of $\frac{\partial f}{\partial r_j}( \boldsymbol{r})  = 0$, for $j \in [G]$. Then
\begin{equation} \label{Eq:Solutionr_j^*_lmm}
r_j^* = N_j \left( 1 + \frac{1}{W_{-1} \left( - e^{-\left(\alpha_{(j)}\mu_{(j)} + 1 \right)} \right) } \right) ,
\end{equation}
where $W_{-1}(x)$\footnote{$W_{-1}(x)$ denotes the branch satisfying $W(x)e^{W(x)} = x$ and $W(x)\le -1$.} is the lower branch of the Lambert $W$ function.
\end{lemma}

\begin{proof}
For $j\in [G]$, $\frac{ \partial f  }{\partial r_j}( \boldsymbol{r}^* ) = 0 $ reduces to
\begin{equation*} 
r_j^*  \frac{\partial \xi}{\partial r_j}(r_j^*, N_j, \mu_{(j)}) - \xi ( r_j^*, N_j, \mu_{(j)}) = 0  .
\end{equation*}
This equation is rephrased as
\begin{equation*} 
\frac{r_j^*}{\mu_{(j)} (N_j - r_j^*)} = \alpha_{(j)} + \frac{1}{\mu_{(j)}} \log \left( \frac{N_j}{N_j - r_j^*} \right)  .
\end{equation*}
Let $z = \frac{N_j}{N_j - r_j^*} > 1$. Then this equation is represented as
\begin{equation*} 
-e^{-(\alpha_{(j)} \mu_{(j)}+1)} = -z e^{-z} = W_{-1}^{-1}(-z) ,
\end{equation*}
where $W_{-1}^{-1}$ denotes the inverse function of $W_{-1}$.
It follows that 
\begin{equation} \label{Eq:Lambert4}
-z = W_{-1} (-e^{-(\alpha_{(j)} \mu_{(j)}+1)})  .
\end{equation}
By solving (\ref{Eq:Lambert4}) with respect to $r_j^*$, we get the result in (\ref{Eq:Solutionr_j^*_lmm}).
\end{proof}

The following theorem provides the optimal load allocation and the optimal $(n,k)$ MDS code to achieve the minimum of $\max_{j\in [G]} \left\{ \lambda_{r_j: N_j}^{l_{(j)}} \right\}$.

\begin{theorem} \label{Thm:Main}
The optimal load allocation $(\boldsymbol{l}^*, \boldsymbol{r}^*)$ to achieve the minimum of $\max_{j\in [G]} \left\{ \lambda_{r_j: N_j}^{l_{(j)}} \right\}$, denoted by $T^{\star}$, is determined as follows:

\begin{align} \label{Eq:Solutionr_j^*}
&r_j^* = N_j \left( 1 + \frac{1}{W_{-1} \left( - e^{-\left(\alpha_{(j)}\mu_{(j)} + 1 \right)} \right) } \right)
\end{align}
and
\begin{align} \label{Eq:Solutionl_j^*} 
l_{(j)}^* = \frac{k} {r_j^* + \sum_{j' \neq j}  r_{j'}^* \frac{ \xi ( r_j^*, N_j, \mu_{(j)} ) }{\xi \left( r_{j'}^*, N_{j'}, \mu_{(j')} \right) } } , 
\end{align}
for $j \in [G]$, where
\begin{align} \nonumber
& \xi(r_j^*,N_j,\mu_{(j)}) = \alpha_{(j)} + \frac{1}{\mu_{(j)}} \log \left( - W_{-1} \left( -e^{-(\alpha_{(j)} \mu_{(j)} + 1)} \right) \right) \\ \label{Eq:r_jstarxi}
\textnormal{ and } \hspace{0.5in}& \frac{r_j^*}{\xi \left(r_j^*, N_j, \mu_{(j)} \right) } = - \frac{\mu_{(j)} N_j }{W_{-1} (- e^{-(\alpha_{(j)}\mu_{(j)} + 1 )})}  .
\end{align}

Furthermore, for given $k$, we have the optimal $(n^*, k)$ MDS code, where $n^* = \sum_{j\in [G]} N_j l_{(j)}^*  .$
Then the minimum expected latency, $T^{\star}$, is represented as
\begin{equation} \label{Eq:Optimal_latency_lower}
- \frac{1}{ \sum_{j\in [G]} \frac{\mu_{(j)} N_j }{W_{-1} \left( - e^{- \left( \alpha_{(j)}\mu_{(j)} + 1 \right) } \right) } }  .
\end{equation}
\end{theorem}

\begin{figure*}[t]
\centering
\begin{multicols}{2}
\includegraphics[width=0.97\linewidth]{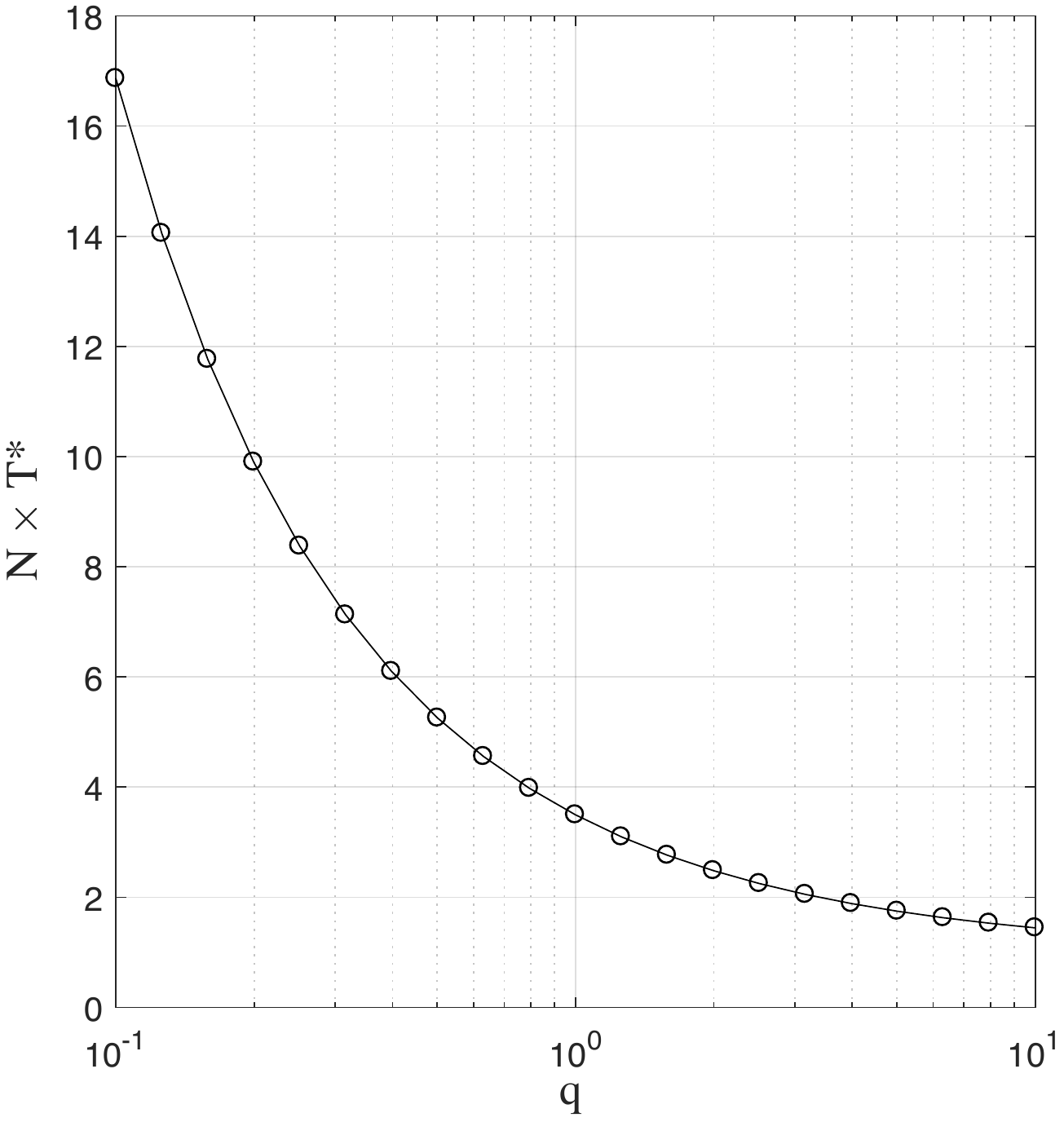} \par\caption{$N \times T^{\star}$ as a function of $q \boldsymbol{\mu}$, where $q$ is the scale of $\boldsymbol{\mu}$.}
	\label{Fig:Tstar_vs_mu}
\includegraphics[width=1\linewidth]{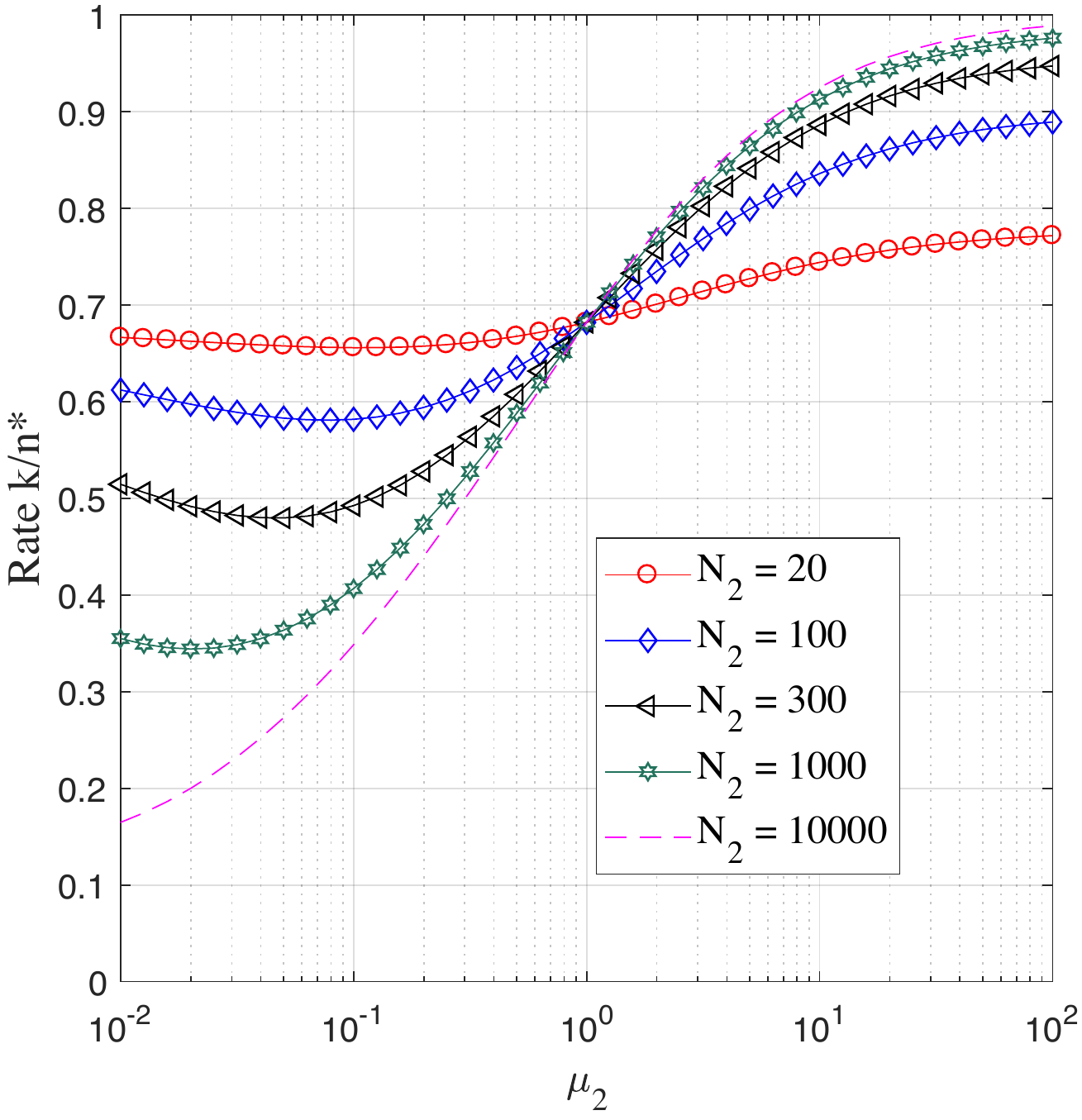} \par\caption{Rates of MDS code with the proposed load allocation for fixed $(N_1 = 100, \mu_{(1)} = 1,\alpha_{(1)} = 1)$  and various $(N_2, \mu_{(2)},\alpha_{(2)} = 1)$.}
	\label{Fig:Rate_fixed_N1}
\end{multicols}
\vspace{-0.1in}
\end{figure*}


\begin{proof}
Recall that the optimal load allocation $(\boldsymbol{l}^*, \boldsymbol{r}^* )$ satisfies the constraints (\ref{Eq:Sumk}) and (\ref{Eq:Mean_equality_condition}).
By solving the (\ref{Eq:Mean_equality_condition}) with respect to $l_{(j')}$, we have
\begin{equation} \label{Eq:Added_condition_refined}
	l_{(j')} = l_{(j)} \frac{\xi(r_j, N_j, \mu_{(j)})}{ \xi (r_{j'}, N_{j'}, \mu_{(j')})} \hspace{0.5in} \textnormal{ for } j'\neq j \in [G]  .
\end{equation}
Inserting  (\ref{Eq:Added_condition_refined}) to (\ref{Eq:Sumk}), we get 
\begin{equation} \label{Eq:loadl_j}
l_{(j)} = \frac{k} {r_j + \sum_{j' \neq j} r_{j'} \frac{ \xi ( r_j, N_j, \mu_{(j)} )}{ \xi ( r_{j'}, N_{j'}, \mu_{(j')} )}}  .
\end{equation}
It follows from (\ref{Eq:Mean_equality_condition}) and (\ref{Eq:loadl_j}) that 
\begin{equation} \label{Eq:Solution_G}
\max_{j\in [G]} \left\{ \lambda_{r_j: N_j}^{l_{(j)}} \right\} =   f ( \boldsymbol{r})  . 
\end{equation}

It follows from Lemma \ref{Lmm:Convexity_f(r)} that $f (\boldsymbol{r})$ has the unique extreme point $\boldsymbol{r}^*$ on $S$ defined in (\ref{Eq:CartesianS}).
From Lemma \ref{Lmm:r_j^*proof}, we have $r_j^*$ as in (\ref{Eq:Solutionr_j^*}) for $j \in [G]$. By inserting (\ref{Eq:Solutionr_j^*}) to (\ref{Eq:loadl_j}), we get $l_j^*$ described in (\ref{Eq:Solutionl_j^*}) for $j \in [G]$.
By inserting (\ref{Eq:Solutionr_j^*}) to (\ref{Eq:Solution_G}), we obtain
\begin{equation*}
f ( \boldsymbol{r}^*) = \frac{1}{ \sum_{j\in [G]} \frac{r_j^*}{\xi \left(r_j^*, N_j, \mu_{(j)} \right) }}  .
\end{equation*}
We obtain (\ref{Eq:r_jstarxi}) using the following equality $\log (-W_{-1} (z) ) + W_{-1}(z) = \log ( - z )  ,$
where $W_{-1}(z)\le -1$ and $z \in \left[-\frac{1}{e}, 0\right)$.
Therefore, we obtain the minimum expected latency in (\ref{Eq:Optimal_latency_lower}).
\end{proof}


The load allocation in \eqref{Eq:Solutionl_j^*} is real value. Accordingly, for implementation, we convert the result $l_{(j)}^*$ to integer $\ceil{l_{(j)}^*}$. Here, we use the ceil function.
In practical scenarios, the number of rows of $\boldsymbol{A}$, $k$, is fairly large, for example $k$ ranges from hundreds of thousands to millions. This implies that the number of rows assigned to worker is in the order of hundreds to thousands. Therefore, the round function on the optimal load allocation has a negligible effect on the performance.


Note that $T^{\star}$  depends only on $\boldsymbol{\mu}$ and $\boldsymbol{N}$. 
It can be easily seen that $T^{\star} = \Theta(\frac{1}{N})$. Fig. \ref{Fig:Tstar_vs_mu} illustrates the above statement where $\boldsymbol{N}:= (N_1, N_2, N_3) =  (1000, 2000, 3000)$, $\boldsymbol{\mu} = (\mu_{(1)}, \mu_{(2)}, \mu_{(3)}) = (2, 1, 0.5)$, and $\boldsymbol{\alpha} = (\alpha_{(1)}, \alpha_{(2)}, \alpha_{(3)}) = (1, 1, 1)$. We denote the scale of $\boldsymbol{\mu}$ as $q$.
In addition, $\frac{k}{n^*}$, the rate of $(n^*,k)$ MDS code, is a function of $\boldsymbol{\mu}$ and $\boldsymbol{N}$. 
%
Next, we observe how the straggling parameter and the number of workers in each group influence the rate. It is assumed that there are two groups in order to facilitate visualization. In Fig \ref{Fig:Rate_fixed_N1}, $N_1$ and $\mu_{(1)}$ are set to $100$ and $1$, respectively. In addition, we set $\alpha_{(1)} = \alpha_{(2)} = 1$ to observe the change in the rate $\frac{k}{n^*}$ as the value of $(N_2,\mu_{(2)})$ varies.
If there is only one group, then the rate $\frac{k}{n^*}$ is a strictly increasing function with respect to straggling parameter.
Intuitively, for fixed $N_2$, the rate can be thought of as a strictly increasing function with respect to straggling parameter $\mu_{(2)}$. Interestingly, however, it is not true, as we can see in Fig. \ref{Fig:Rate_fixed_N1}. Next, we show that Theorem \ref{Thm:Main} is a generalization of the result in  \cite{Lee18Speeding}.


\begin{remark}\label{remark: gen}
Consider the case that there exist two groups. Assume that the straggling parameter of workers in the second group is arbitrarily small, i.e., $\mu_{(2)} \approx 0$. Note that $\xi(r_2, N_2, \mu_{(2)})$ in \eqref{Eq:xi} goes to infinity as $\mu_2$ goes to zero. This gives that
$\lim_{\mu_{(2)} \rightarrow 0} f( \boldsymbol{r}) = \frac{1}{r_1} \xi(r_1, N_1, \mu_{(1)})  .$

In the following way we reach the same conclusion as well.
Assume that there are $N_j$ workers in group $j$ and all workers have the same straggling and shift parameters, i.e., $\mu_{(j)} = \mu$ and $\alpha_{(j)} = \alpha$ for all $j \in [G]$.
From \eqref{Eq:Solutionl_j^*} and \eqref{Eq:Optimal_latency_lower}, we have
\begin{equation*}
l_{(j)}^* = \frac{k}{N \left(1 + \frac{1}{W_{-1} (-e^{-(\alpha \mu + 1)})} \right)}
\end{equation*}
and 
\begin{equation*}
T^{\star} = - \frac{W_{-1} \left( -e^{- (\alpha \mu + 1) }\right)}{\mu N}  .
\end{equation*}
These observations tell us that the result in our setting is a generalization of the result in \cite{Lee18Speeding} considering only one group with the same straggling and shift parameters (i.e, homogeneous workers).
\end{remark}


\subsection{Asymptotic Behavior of $\lambda_{r:N}$} \label{Subsec:Analysis_lambda}
In Section \ref{Subsec:Loadallocation}, we get the optimal load allocation for the lower bound of the expected latency $\lambda_{r:N}$ given $k$, $\boldsymbol{N}$, and $\boldsymbol{\mu}$. In this subsection, we show that for given $\boldsymbol{l}^*$ and $\boldsymbol{r}^*$ obtained from Theorem \ref{Thm:Main}, $\lambda_{r:N}$ converges to $T^{\star}$ as $N$ goes to infinity.
First, we introduce the asymptotic results for order statistic.


\begin{proposition} \label{Lmm:Momentum} [Theorem 6.1.1 in  \cite{Reiss89Approximate}, (central order statistic)] 
Let $l_{(j)}$ be given. For fixed $0 < q_j < 1$,
$T_{q_jN_j : N_j}^{l_{(j)}}$ converges in distribution to $X$, where 
\begin{equation} \label{Eq:Central_order}
X \sim \mathcal{N} \left( \eta_j,  \sigma_{j}^2  \right) \textnormal{, } \hspace{0.5in} \sigma_{j}^2 = \frac{q_j(1-q_j)}{N_j (f_j(\eta_j))^2} ,
\end{equation}
$F_j' = f_j$, and $\eta_j = F_{j}^{-1}(q_j)$. We denote $T_{q_jN_j : N_j}^{l_{(j)}} \xrightarrow[]{d} \mathcal{N} \left( \eta_j, \sigma_{j}^2  \right)  .$
\end{proposition}

Next, for given $\boldsymbol{l}^*$ and $\boldsymbol{r}^*$, we show that $\sigma_j^2$ goes to zero as $N$ goes to infinity in Lemma \ref{Lmm:Sigma_j=0}. 
Finally, using Lemma \ref{Lmm:Sigma_j=0}, we will prove that $\lambda_{r:N}$ converges to $T^{\star}$ in Theorem \ref{Thm:Convergencelambda}, if $\boldsymbol{l}^*$ and $\boldsymbol{r}^*$ are given.
For simplicity, we denote
\begin{align*}
\bar{\xi}(\mu_{(j)}) \coloneqq \xi(r_j^*,N_j,\mu_{(j)})  =  \alpha_{(j)} + \frac{1}{\mu_{(j)}} \log \left( - W_{-1} \left( -e^{-(\alpha_{(j)}\mu_{(j)} + 1)} \right) \right)  .
\end{align*}

\begin{lemma} \label{Lmm:Sigma_j=0}
For given ($\boldsymbol{l}^*, \boldsymbol{r}^*$), the variance $\sigma_j^2$ in (\ref{Eq:Central_order}) goes to zero as $N$ goes to infinity.
\end{lemma}
\begin{proof} \label{Eq:Sigma1}
Using (\ref{Eq:Solutionr_j^*}), we obtain $q_j^* = \frac{r_j^*}{N_j} = 1 + \frac{1}{W_{-1} \left( - e^{-\left(\alpha_{(j)}\mu_{(j)} + 1 \right)} \right) } .$
Recall that 
\begin{equation} \label{Eq:Sigma2}
f_j (\eta_j^*) =  \mu_{(j)} \frac{k }{l_{(j)}^*} e^{-\mu_{(j)} \left( \frac{k }{l_{(j)}^*} \eta_j^* - 1 \right)}  .
\end{equation}
From (\ref{Eq:Solutionl_j^*}), we have
\begin{equation}\label{Eq:Sigma3}
\frac{k}{l_{(j)}^*} = N_j q_j^*+ \sum_{j' \neq j}   N_{j' } q_{j'}^* \frac{\bar{\xi}(\mu_{(j)}) }{\bar{\xi}(\mu_{(j')}) }  .
\end{equation}
Note that 
\begin{equation} \label{Eq:Sigma4}
\eta_j^* = F_j^{-1} (q_j^*) = \frac{l_{(j)}^*}{k} \bar{\xi}(\mu_{(j)})  .
\end{equation}
By inserting  (\ref{Eq:Sigma3}) and (\ref{Eq:Sigma4}) to (\ref{Eq:Sigma2}), we obtain
\begin{equation*}
f_j( \eta_j^* )  = \left ( N_j q_j^*+ \sum_{j' \neq j}   N_{j' } q_{j'}^* \frac{\bar{\xi}(\mu_{(j)}) }{\bar{\xi}(\mu_{(j')}) }\right) a(\mu_{(j)})  , \nonumber
\end{equation*}
where $a(\mu_{(j)}) = - \frac{\mu_{(j)}}{ W_{-1} (-e^{-(\alpha_{(j)} \mu_{(j)} + 1)})}  .$
Observe that $q_{j'}^*$, $\bar{\xi}(\mu_{(j)})$, and $a(\mu_{(j)})$ depend only on $\mu_{(j)}$.
In addition, if $N$ goes to infinity, then $N_j$ goes to infinity for $j \in [G]$.
Thus, $\sigma_j^2$ converges to zero as $N$ goes to infinity for given  ($\boldsymbol{l}^*,\boldsymbol{r}^*$).
\end{proof}

Lemma \ref{Lmm:Sigma_j=0} implies that  $T_{r_j^* : N_j}^{l_{(j)}^*} \xrightarrow[]{d} \eta_j^*.$


\begin{theorem} \label{Thm:Convergencelambda}
For given ($\boldsymbol{l}^*, \boldsymbol{r}^*$), 
$\lambda_{r^*: N}$ converges to $T^{\star}$ as $N$ goes to infinity.
\end{theorem}

\begin{proof}
It suffices to show that $T_{r:N} \xrightarrow[]{d}  T^{\star}.$
Since $T_{r_j:N_j}$'s are independent, by using (\ref{Eq:Order_max}), we get $$ F(t) := \Pr(T_{r:N}\le t)  = \prod_{j\in [G]} \Pr \left(T_{r_j:N}^{l_{(j)}}\le t \right) .$$
It follows from Proposition \ref{Lmm:Momentum} and Lemma \ref{Lmm:Sigma_j=0} that for given ($\boldsymbol{l}^*, \boldsymbol{r}^*$), we have  $$\prod_{j\in [G]} \Pr \left(T_{r_j:N}^{l_{(j)}}\le t \right) \rightarrow \prod_{j\in [G]} H_{\eta_j^*}(t) ,$$
where $H_{a}(t) = 1$ if $t \ge a$ and $H_{a}(t) = 0$ if $t < a$.
Since $\eta_j^* = T^{\star}$ for all $j\in [G]$. we get $$\prod_{j\in [G]} H_{\eta_j}(t) = H_{T^{\star}}(t)  .$$
It follows that $T_{r:N}$ converges in distribution to $T^{\star}$.
\end{proof}

Recall that $\lambda_{r:N} \ge  \max_{j\in [G]} \left\{ \lambda_{r_j: N_j}^{l_{(j)}} \right\} \ge \lambda_{r_j^*: N_j}^{l_{(j)}^*} = T^{\star}  . $
Note that the quantity $\lambda_{r:N}$ does not mean the maximum of the expected latencies of groups required to wait for $r_j$ workers from each group $j$ for $j\in [G]$. 
$\lambda_{r:N}$ means the expected latency when the master aggregates $k$ coded inner products from nonidentically independent $N$ workers. 
Since the load allocation in \eqref{Eq:Solutionl_j^*} is obtained to achieve the minimum of $ \max_{j\in [G]} \left\{ \lambda_{r_j: N_j}^{l_{(j)}} \right\}$, the quantity $\lambda_{r:N}$ with the  load allocation in \eqref{Eq:Solutionl_j^*} can be thought of as greater than $T^{\star}$.
However, it follows from the result in the subsection that we obtain the load allocation in $\eqref{Eq:Solutionl_j^*}$ is also (asymptotically) optimal for $\lambda_{r:N}$  thanks to the asymptotic behavior of central order statistic.
In short, for given $(\boldsymbol{l}^*, \boldsymbol{r}^*)$, Theorem \ref{Thm:Convergencelambda} indicates that $\lambda_{r:N }$ converges to the lower bound $T^{\star}$ as $N$ goes to infinity.

\subsection{Uniform Load Allocation} \label{Subsec:Uniform}

In contrast to the aforementioned load allocation strategy, in this subsection, we simply consider the uniform load allocation, i.e., we assign the same amount of task to each worker. For fixed $k$, $\boldsymbol{N}$, and $\boldsymbol{\mu}$, we consider the following constraints:
\begin{align*} 
n = \sum^N_{j = 1} N_j l_{(j)} \,  \hspace{0.5in} \textnormal{and}  \hspace{0.5in} \sum_{j \in [G]} r_j = r  .
\end{align*}
In the above constraints, we consider the two cases where $n$ is fixed or $r$ is fixed.
Uniform load allocation schemes for fixed parameters $n$ and $r$ are presented as follows.

\

\subsubsection{Uniform load allocation for given $n$}
Assume that $n$ is fixed. This means that we have an $(n,k)$ MDS code.
Since we assume that $l_{(j)}$ is constant, we obtain $l_{(j)}N = n$ for $j\in [G]$. This gives that $l_i = \frac{n}{N}$ for $i \in [N]$. Let us denote $l^{u} := \frac{n}{N}$.
In addition, the condition \eqref{Eq:Sumk} for guaranteeing the successful recovery of $\mtx{A}\mathbf{x}$ is rephrased as
\begin{equation} \label{Eq:Constraint_group_uniform}
\sum_{j \in [G]} r_j = \frac{kN}{n}  . 
\end{equation}
Note that (\ref{Eq:Constraint_group_uniform}) implies that $\frac{k}{n} = \frac{r}{N}$, where $\frac{k}{n}$ is a code rate of an $(n,k)$ MDS code.

Consider the $(n^*,k)$ MDS code obtained through the optimal load allocation in Theorem \ref{Thm:Main}. 
When using the $(n^*,k)$ MDS code, the expected latency of the proposed load allocation and the uniform load allocation will be compared.

\subsubsection{Uniform load allocation for given $r$} \label{Subsubsec:Uniform}
The system model used in \cite{Kim19Coded} is described according to our probabilistic model as follows. \footnote{Here, we set $\alpha_{(j)} = \alpha$ for $j\in [G]$ to apply for the scheme proposed in \cite{Kim19Coded}.}
In our system model, $r$ is a parameter created by assuming the time to receive $k$ inner products required for MDS decoding. However, in \cite{Kim19Coded}, it is assumed that a fixed $r$ is given. The assumption that $r$ is fixed means that the master performs MDS decoding to obtain $\mtx{A} \mathbf{x}$ when the computation results are received from $r$ workers out of $N$ workers. Thus, in order to use this system model, 
the total number of workers is required to be greater than or equal to $r$. Then we have $l^r := l_i = \frac{k}{r} $ for $i \in [N] .$
Thus, for fixed $r$, we use an $(n,k)$ MDS code, where $n = N l^r $.
It can be interpreted that the number of rows of the data matrix allocated to each worker is fixed to $l^r$.

Regarding this case, we provide the following observation.
For a fixed $r$, let $ (r_1, r_2, \dots, r_G)$ be given.
For group $j$, the expected latency is then represented as
$
\lambda_{r_j: N_j} = \frac{l^r}{k} \left( \alpha + \frac{1}{\mu_{(j)}} \log \left( \frac{N_j}{N_j - r_j}\right) \right)  ,
$
which implies that $\lambda_{r_j: N_j}$ converges to $\frac{1}{r}$ as $N$ approaches infinity.
Thus, for any $ (r_1, r_2, \dots, r_G)$, we have $\lambda_{r_j:N_j}$ is equal to $\frac{1}{r}$ for a sufficiently large $N$.
This means that the expected latency of a system using the MDS code with a fixed $r$ is given by $\frac{1}{r}$ for a sufficiently large $N$.

In \cite{Kim19Coded}, the authors proposed a scheme to achieve the optimal latency of the MDS code described above.
Then, we get the following corollary to Lemma \ref{Lmm:Lower_conditionG=2} for fixed $r$.

\begin{corollary}\label{Cor:Lower_condition_uniform}
For given $G = 2$ and $n$, consider the uniform load allocation, i.e. $l_i = l^r$. Let $\boldsymbol{r} = (r_1, r_2)$ be the numbers of workers to complete the uniformly assigned task $l^r$  which achieves the minimum of $ \max_{j\in [G]} \left\{ \lambda_{r_j: N_j}^{l^{r}} \right\}$.
Then 
\begin{equation} \label{Eq:AddedConstraint_uniform}
\hspace{0.5in} \lambda_{r_j: N_j}^{l^r} = \lambda_{r_{j'}: N_{j'}}^{l^r}  \,, \hspace{0.5in} \textnormal{ for all }  j \neq j' \in [G].
\end{equation}
\end{corollary}
\begin{proof}
This corollary can be easily verified by applying the same method used in Lemma \ref{Lmm:Lower_conditionG=2}. So we omit the proof.
\end{proof}

In the uniform load allocation, we have $\lambda_{r_j:N_j}^{l^r} =  \frac{l^r}{k} \left( \alpha + \frac{1}{\mu_{(j)}} \log \left(\frac{N_j}{N_j - r_j} \right) \right) .$
It follows from (\ref{Eq:AddedConstraint_uniform}) that we have
\begin{equation} \label{Eq:Mean_equality_condition_uniform}
\frac{1}{\mu_{(j)}} \log \left(\frac{N_j}{N_j - r_j} \right) = \frac{1}{\mu_{(j')}} \log \left(\frac{N_{j'}}{N_{j'} - r_{j'}} \right)  .
\end{equation}
Solving the (\ref{Eq:Mean_equality_condition_uniform}) with respect to $r_{j'}^u$, we insert $r_{j'}^u$ to (\ref{Eq:Constraint_group_uniform}).
Then we obtain the following theorem.
\begin{theorem} [Theorem 3 in \cite{Kim19Coded}] \label{Thm:group_coded_equation}
For given $r$ and $G$ $(= 2)$, consider the uniform load allocation. Then $\boldsymbol{r}$ can be determined by solving the following equation.
\begin{equation} \label{Eq:group_coded_equation}
\hspace{0.5in} r_j + \sum_{j' \neq j} N_{j'} \left( 1- \left( 1- \frac{r_j}{N_j} \right)^{\frac{\mu_{(j')}}{\mu_{(j)}}} \right) = r \,,  \hspace{0.5in} \textnormal{ for } j\in [G]  .
\end{equation}
\end{theorem}

The proposed scheme in \cite{Kim19Coded} starts with dividing the data matrix $\mtx{A}$ with $k$ row vectors into $r$ submatrices of the same size. Then $r_j$, the number of submatrices assigned to group $j$, is obtained from (\ref{Eq:group_coded_equation}). The $(N_j, r_j)$ MDS code is used to encode the submatrices assigned to each group and send it to the workers in group $j$, for $j \in [G]$.\footnote{In our modeling, we use an $(n,k)$ MDS code for the entire matrix.} After the master collect $r_j$ computation results from group $j$, the master performs decoding to recover the original $r_j$ submatrices. 

In general, (\ref{Eq:group_coded_equation}) may not have a solution, if $G > 2$.
For example, If $G = 3$, $r =200$, $(N_1, N_2, N_3) = (100, 200, 300)$, and $(\mu_{(1)}, \mu_{(2)}, \mu_{(3)}) = (3, 2, 1)$, then there is no solution.

\subsection{Load Allocation Scheme in \cite{Reisizadeh19Coded}}
In \cite{Reisizadeh19Coded}, the authors proposed a load allocation method. In this section, we propose the optimal load allocation method according to the probabilistic model considered in \cite{Reisizadeh19Coded}, which is represented as follows.
The cumulative distribution function of execution time of worker $i$ with a shift parameter, $a_i$, is denoted by
\begin{equation} \label{Eq:Reisizadeh}
\hspace{0.5in} F_{i}^{b}(t) = \Pr(T_i^b \le t) = 1-e^{-\frac{\mu_i}{l_i} (t - \alpha_i l_i)} \,, \hspace{0.5in} \textnormal{ for } t \ge \alpha_i l_i
\end{equation}

We use the same system parameters ($G$, $n$, $k$, $l_{(j)}$, $\alpha_{(j)}$, $\mu_{(j)}$, $N_j$, $r_j$) except for the probabilistic model.
Recall that $r_i$ denotes the number of workers in group $j$ which complete the assigned subtasks when the master receives the $k$ inner products, i.e., $\sum_{j\in [G]} l_{(j)}r_j = k$. 
The execution time distribution of worker $i$ in group $j$ assigned to calculate the inner product of $l_{(j)}~(< k)$ rows of $\tilde{\mtx{A}}$ with $\mathbf{x}$
is expressed as
\begin{equation} \label{Eq:Reisizadeh_group}
F_{j}^{b}(t)  = 1-e^{-\frac{\mu_{(j)}}{l_{(j)}} (t - \alpha_{(j)} l_{(j)})}
\end{equation}
for $t \ge \alpha_{(j)} l_{(j)}.$
Let $\lambda_{r:N}^{b}$ denote $E[T_{r:N}^{b}]$ where $r = \sum_{j\in [G]} r_j$ and $T_{r:N}^{b}$ is the $r$-th order statistic of the $N$ exponential random variables with the distribution function in (\ref{Eq:Reisizadeh_group}). Let $\lambda_{r_j:N_j}^{(b, l_{(j)})}$ also denote $E \left [ T_{r_j:N_j}^{(b, l_{(j)})} \right ]$ where  $T_{r_j:N_j}^{(b, l_{(j)})}$ is the $r_j$-th order statistic of the $N_j$ identical exponential random variables with the distribution function in (\ref{Eq:Reisizadeh_group}) and a shift parameter $\alpha_{(j)}$ and a load allocation $l_{(j)}$.
Then we have $\lambda_{r_j:N_j}^{(b, l_{(j)})} = l_{(j)} \left(\alpha_{(j)} + \frac{1}{\mu_{(j)}} \log \left(\frac{N_j}{N_j - r_j} \right) \right) .$


Note that for given $\alpha_{(j)}$, we can apply the same argument used in our probability model. 
Then we have the following corollary to Theorem \ref{Thm:Main}.

\begin{corollary}
The optimal load allocation $(\boldsymbol{l}_b^*, \boldsymbol{r}_b^*)$ to achieve the minimum of $\max_{j\in [G]} \left\{ \lambda_{r_j: N_j}^{(b, l_{(j)})} \right\}$, denoted by $T_b^{\star}$, is determined as follows: 

\begin{align} \nonumber
& r_{b,j}^* = N_j \left( 1+ \frac{1}{W_{-1} ( -e^{-(\alpha_{(j)} \mu_{(j)} +1 )})} \right) \\
 \label{Eq:Solutionl_j_a^*}
\textnormal{ and } \hspace{0.5in} & l_{b,(j)}^* = \frac{k} {r_{b,j}^* + \sum_{j' \neq j}   r_{b,j'}^*\frac{ \xi ( r_{b,j}^*, N_{j}, \mu_{(j)} ) }{ \xi ( r_{b,j'}^*, N_{j'}, \mu_{(j')} )}}  , 
\end{align}
for $j \in [G]$, where $\xi(r_{b,j}^*,N_j,\mu_{(j)}) = \; \alpha_{(j)} + \frac{1}{\mu_{(j)}} \log \left( - W_{-1} \left( -e^{-(\alpha_{(j)}\mu_{(j)} + 1)} \right) \right)$
and $\frac{ r_{b,j}^*}{ \xi ( r_{b,j}^*, N_{j}, \mu_{(j)} )} = -\frac{\mu_{(j)} N_j}{W_{-1} (-e^{-(\alpha_{(j)} \mu_{(j)} + 1 ) })} .$
Furthermore, for given $k$, we have the optimal $(n^*, k)$ MDS code, where $n^* = \sum_{j\in [G]} N_j l_{b,(j)}^*  .$ 
Then the minimum expected latency, $T_b^{\star}$, is represented as
\begin{equation} \label{Eq:Optimal_latency_lower_a}
T_b^{\star} = 
- \frac{k}{ \sum_{j\in [G]} \frac{\mu_{(j)} N_j }{W_{-1} \left( - e^{- \left( \alpha_{(j)}\mu_{(j)} + 1 \right) } \right) }}  .
\end{equation}
\end{corollary}

\begin{proof}
Proof of this corollary can be done by applying the same argument used in Theorem \ref{Thm:Main}, and thus we omit the proof.
\end{proof}

Reapplying the argument used in Section \ref{Subsec:Add_condition} and \ref{Subsec:Analysis_lambda} , we have $\lambda_{r:N}^{b} \ge  \max_{j\in [G]} \left\{ \lambda_{r_j: N_j}^{(b, l_{(j)})} \right\} \ge \lambda_{r_j^*: N_j}^{(b, l_{(j)}^*)} = T_b^{\star}  . $
Moreover, for given $(\boldsymbol{l}_b^*, \boldsymbol{r}_b^*)$, we can show that $\lambda_{r:N }^b$ converges to the lower bound $T_b^{\star}$ as $N$ goes to infinity.

Similarly to Remark \ref{remark: gen}, we can also confirm that our finding on the load allocation and expected latency is a generalization of the
result for homogeneous workers under the latency model with a shift parameter as follows.
Consider the case that there are $N_j$ workers in group $j$ for $j \in [G]$.
Assume that $\mu_{(j)} = \mu$ and $\alpha_{(j)} = \alpha$ for $j \in [G]$.
From \eqref{Eq:Solutionl_j_a^*} and \eqref{Eq:Optimal_latency_lower_a}, we have 
\begin{equation*} \label{Eq:Group1_allocation}
l_{b,(j)}^* = \frac{k}{N \left(1 + \frac{1}{W_{-1} (-e^{-(\alpha \mu + 1)})} \right)}
\end{equation*}
and 
\begin{equation} \label{Eq:Group1}
T_b^{\star} = - \frac{k W_{-1} \left( -e^{- (\alpha \mu + 1) }\right)}{\mu N}  .
\end{equation}

Again, the same conclusion can be drawn from the following approach.
The aforementioned assumption that all workers have the same straggling and shift parameters means that there is only one group, which naturally leads the uniform load allocation for all workers. 
The data matrix $\mtx{A}\in\mathbb{R}^{k \times d}$ is uniformly divided into $r$ submatrices, then we apply an $(N, r)$ MDS code to obtain $N$ coded submatrices.
In this setting, we obtain $\lambda_{r:N} =  \frac{k}{r} \left( \alpha + \frac{1}{\mu} \log \left( \frac{N}{N - r} \right) \right)  .$
%

Note that $\lambda_{r:N}$ is a function of $r$ and has a unique minimum on an open interval $(0,N)$. Using the similar calculation in Lemma \ref{Lmm:r_j^*proof}, we obtain $r^*$ which minimizes $\lambda_{r:N}$ as follows:
$
r^* = N \left( 1 + \frac{1}{W_{-1} \left( -e^{-(\alpha \mu +1)} \right) } \right)  .
$
Moreover, we get
$
\lambda_{r^*:N} = \frac{k W_{-1}(-e^{-(\alpha \mu +1 )})}{\mu N}  .
$
Observe that $T_b^{\star}$ in \eqref{Eq:Group1} is equal to $\lambda_{r^*:N}$. This concludes that our analysis generalizes the analysis done assuming homogeneous workers with the computation time distribution in \eqref{Eq:Reisizadeh}.

The load allocation method proposed in \cite{Reisizadeh19Coded} is located in Appendix \ref{App:Coded}.

\section{Simulation Results}

%
In this section, we provide comparisons between the expected latency of the proposed load allocation and existing schemes. 
We confine our interest to a certain range of $\mu_{(j)}$ based on the following observations.
For a large $\mu_{(j)}$, we have $W_{-1}(-e^{-(\alpha_{(j)} \mu_{(j)} +1 )}) = -\infty .$
Consider the case where all workers have the same straggling parameters, as in \cite{Lee18Speeding}. The expected latency is then expressed as \eqref{Eq:Group1}. That is, given a sufficiently large straggling parameter, $T^{\star}$ becomes infinite which means that the analysis of the expected latency through this model is not appropriate.
We, thus, perform the evaluation only for the range of $\mu_{(j)}  (< 750)$ for $j\in [G]$.



\begin{figure*}[t]
\centering
\begin{multicols}{2}
\includegraphics[width=1.02\linewidth]{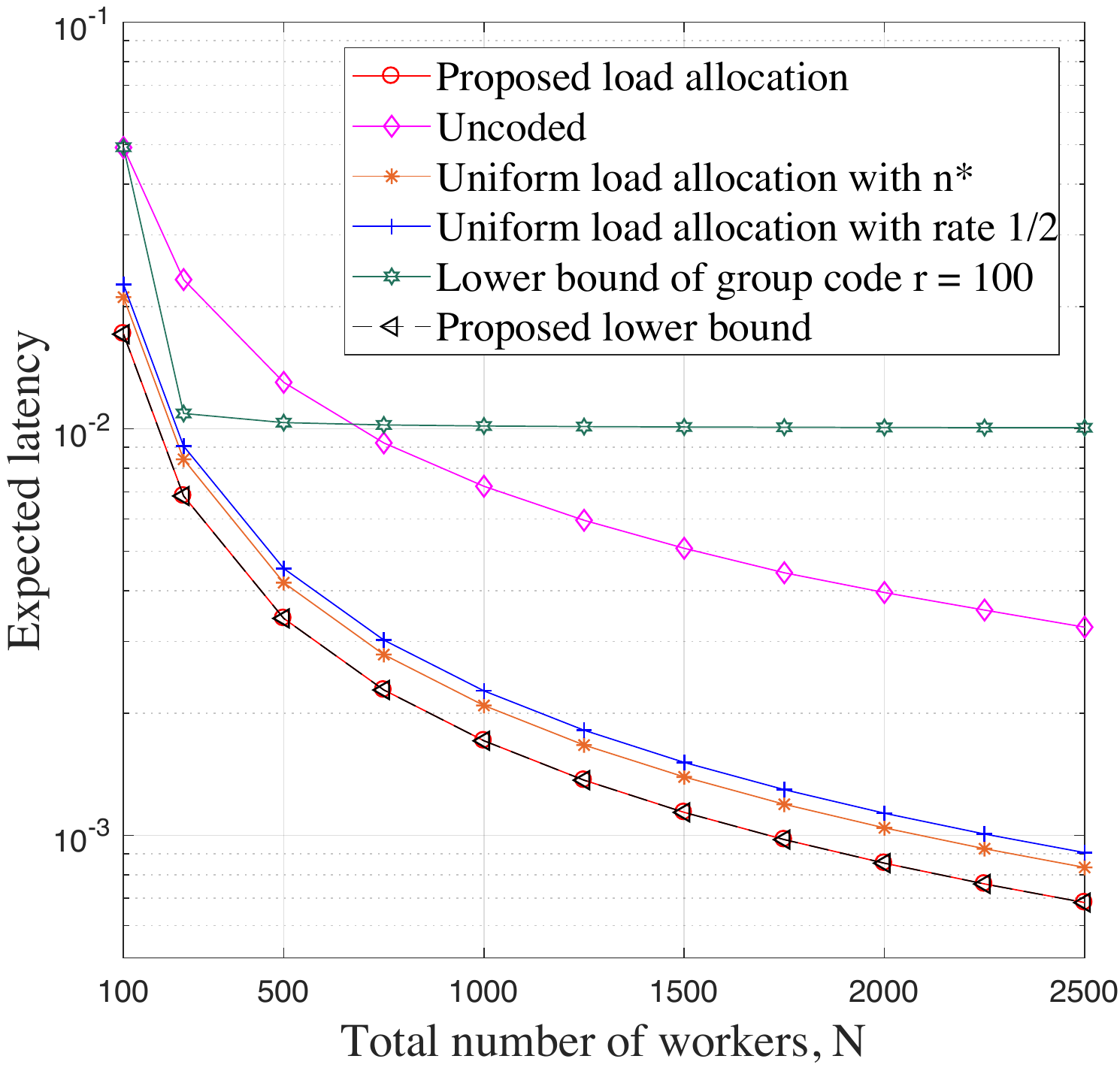} \par\caption{Expected latency comparison between the proposed load allocation, uncoded, uniform load allocation with $n^*$, uniform load allocation with rate $\frac{1}{2}$, the lower bound of group code in \cite{Kim19Coded}, and the proposed lower bound with five groups.}
	\label{Fig:Group5}
\includegraphics[width=0.94\linewidth]{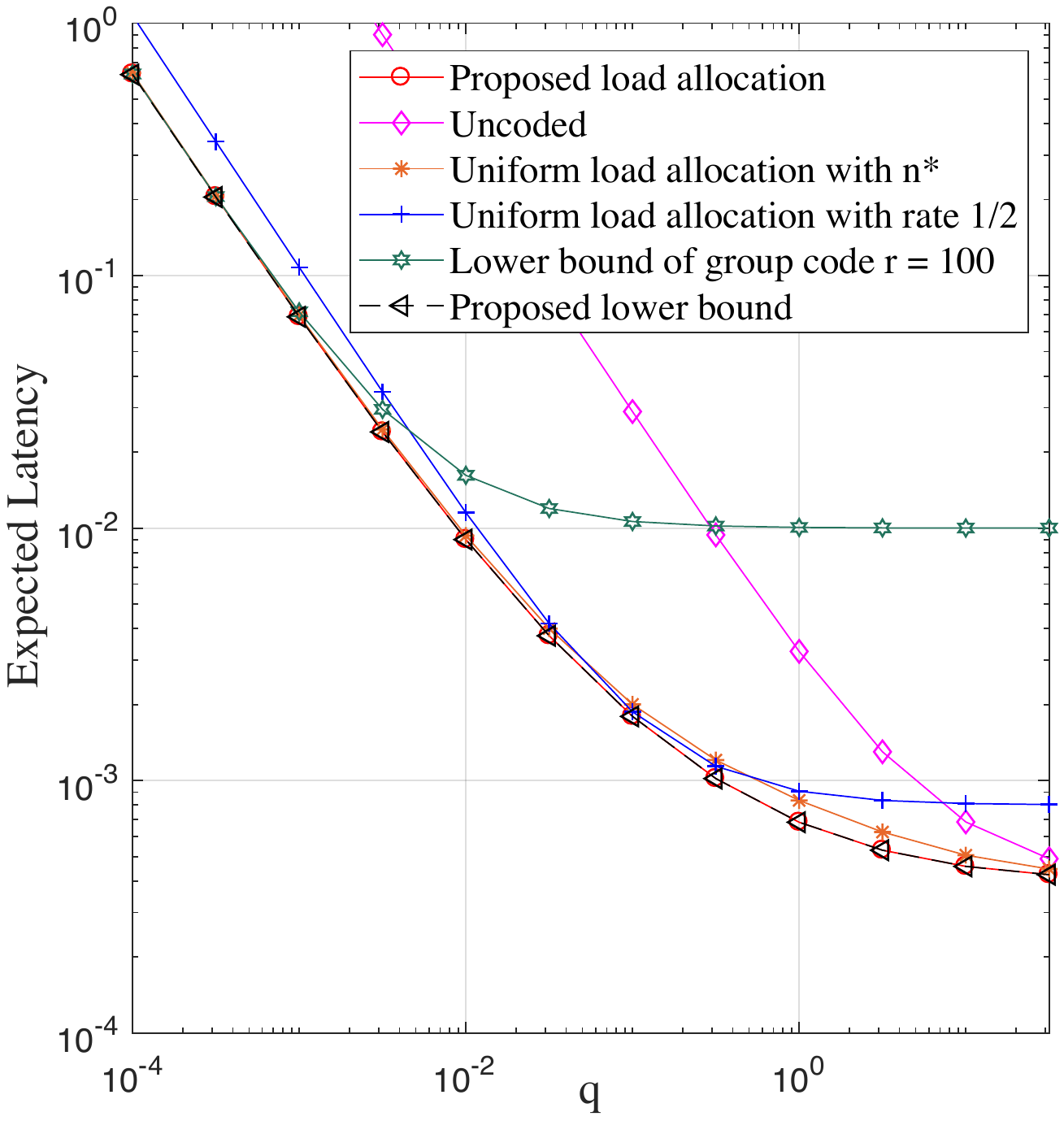} \par\caption{Expected latency comparison between the proposed load allocation, uncoded, uniform load allocation with $n^*$, uniform load allocation with rate $\frac{1}{2}$, the lower bound of group code in \cite{Kim19Coded}, and the proposed lower bound with five groups, according to the scale of $\boldsymbol{\mu}$, denoted by $q$.}
	\label{Fig:Mu_vs_schemes}
\end{multicols}
\vspace{-0.1in}
\end{figure*}

Numerical simulations are carried out using the Monte Carlo method with $10^4$ samples.
In Fig. \ref{Fig:Group5}, we consider the scenario in which workers are formed into five groups.
In Fig. \ref{Fig:Group5}, we set $\boldsymbol{N} = (N_1, N_2, N_3, N_4, N_5) = \frac{1}{25}(3N, 4N,5N,6N,7N) $, $(\mu_{(1)}, \mu_{(2)}, \mu_{(3)}, \mu_{(4)}, \mu_{(5)}  )= (16,12,8,4,1)$, and $r = 100$. In addition, we set $(\alpha_{(1)}, \alpha_{(2)}, \alpha_{(3)},\alpha_{(4)}, \alpha_{(5)}) = (1,1,1,1,1)$ since the scheme in Theorem \ref{Thm:group_coded_equation}  cannot be plotted if $\alpha_{(j)} \neq \alpha_{(j')}$ for $j \neq j' \in [G]$. The result in Fig. \ref{Fig:Group5} shows that the proposed load allocation method achieves the lower bound $T^{\star}$. 
The proposed load allocation shows a 10x or more performance gain over the MDS code with fixed $r$ which is considered as a lower bound in \cite{Kim19Coded} as $N$, the total number of workers, increases.
The reason why the expected latency with fixed $r$ converges to some value (at $10^{-2} = \frac{1}{r}$ in Fig. \ref{Fig:Group5}) without decreasing despite the increase of $N$ is that the load allocation $l_{(j)}$ is constant $\frac{k}{r}$. In other words, even though the total number of workers increases, the amount of task assigned to each worker is constant. It follows that the expected latency cannot continue to decrease.

Next, we compare the proposed load allocation with a uniform load allocation method. For uniform load allocation, $n$ has to be selected first. The simulation in Fig. \ref{Fig:Group5},  is performed with $n$ as $n^*$ and $2k$. 
The uncoded scheme, which takes into account the case where $n = k$, is the uniform load allocation.
Despite using the same $(n^*, k)$ MDS code, the proposed load allocation method has a 18\% lower latency than the uniform load allocation does.

\begin{figure*}[t]
\centering
\begin{multicols}{2}
\includegraphics[width=1.0\linewidth]{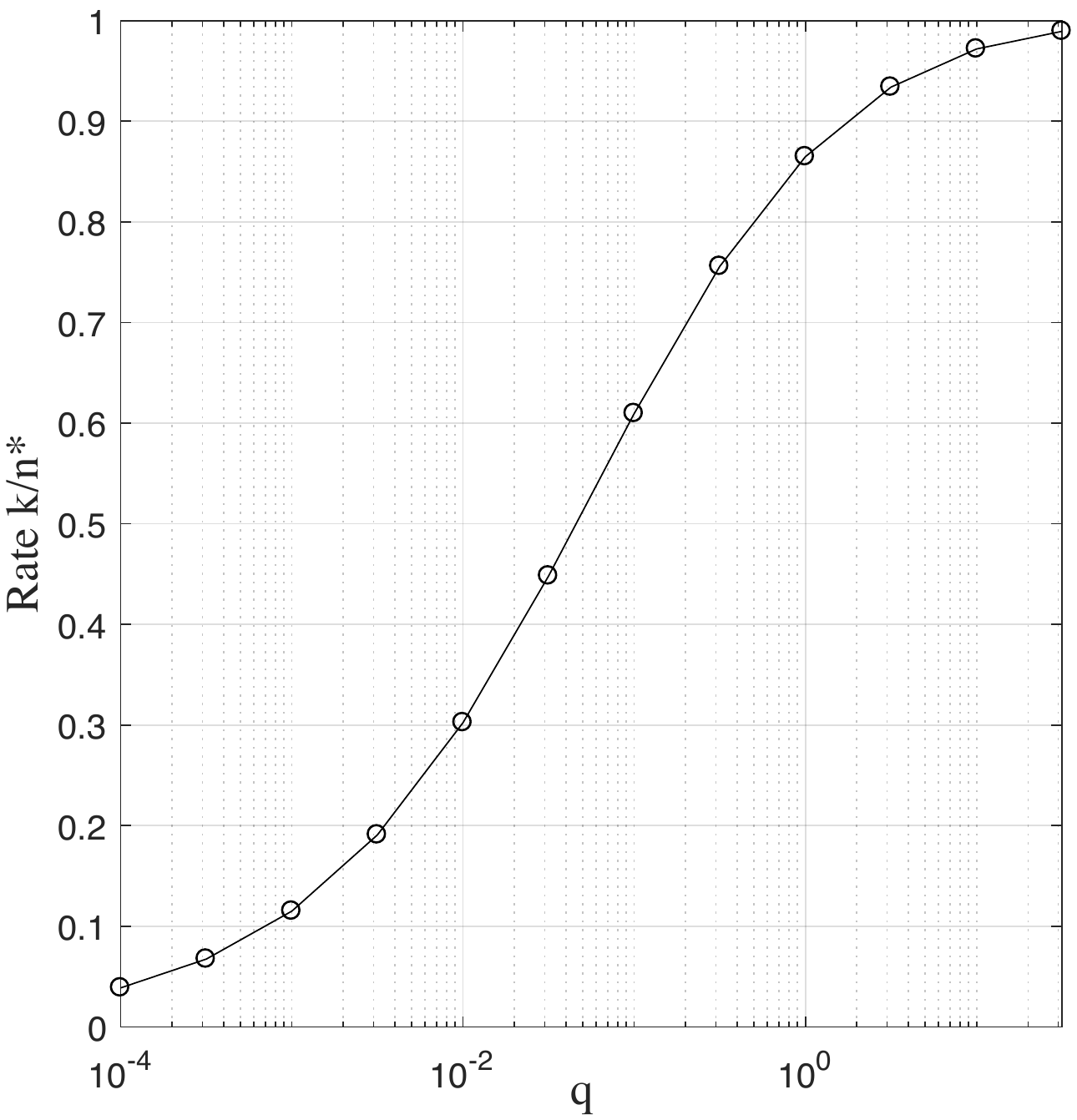} \par\caption{Rate as a function of $q$ (the scale of $\boldsymbol{\mu}$).}
	\label{Fig:Rate_proposed}
\includegraphics[width=1\linewidth]{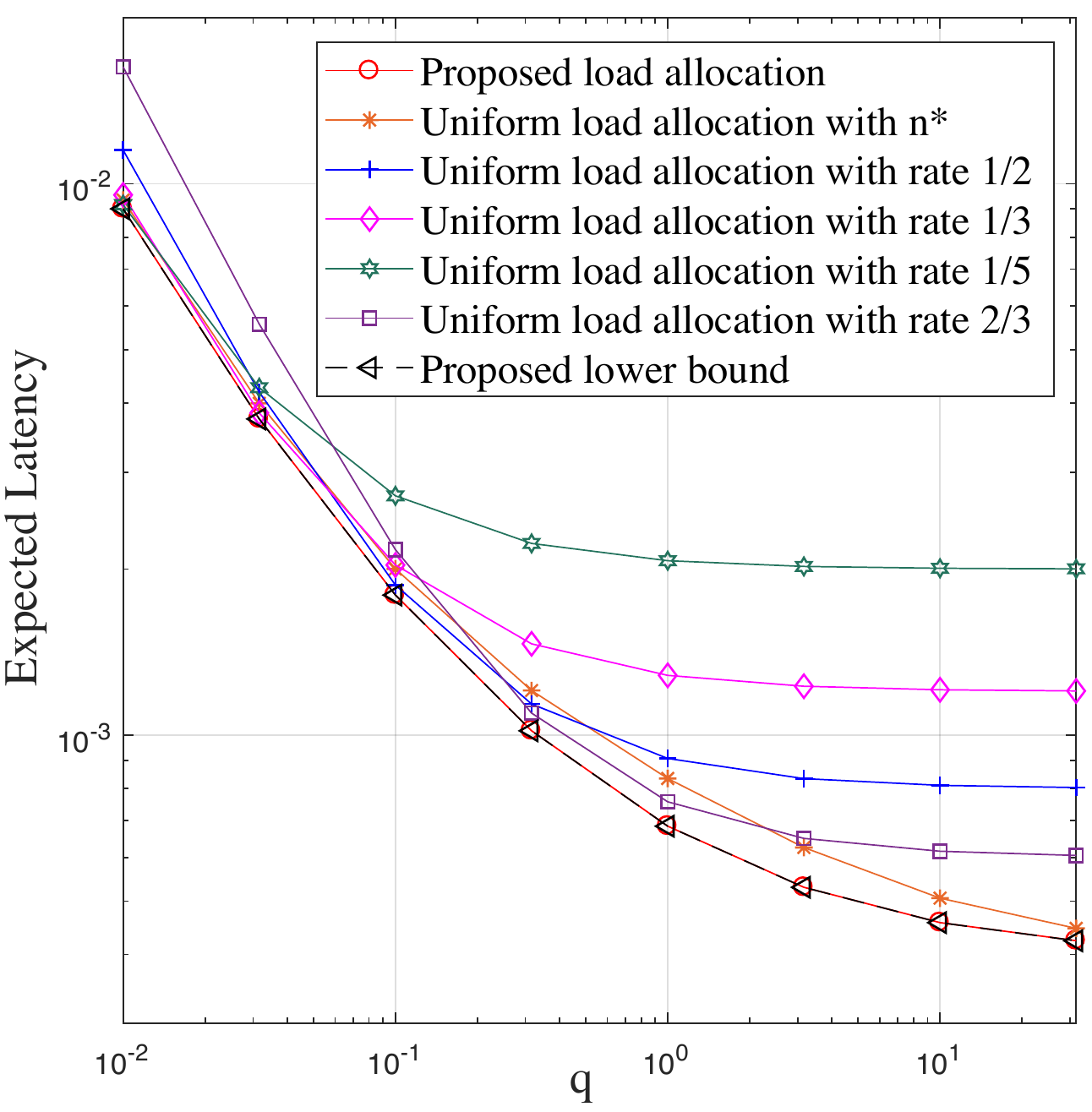} \par\caption{Expected latency  between the proposed load allocation and uniform load allocation with various rates according to $q$ (the scale of $\boldsymbol{\mu}$).}
	\label{Fig:Proposed_vs_uniform}
\end{multicols}
\vspace{-0.1in}
\end{figure*}

Another numerical simulation is conducted to observe the effect of the scale factor of $\boldsymbol{\mu}$, $q$, on the expected latency of each scheme and to find the lowest limit on the expected latency of the uniform load allocation with an $(n,k)$ MDS code.
In this simulation, the same setting as the previous simulation in Fig. \ref{Fig:Group5} is used, except $N$ which is fixed to 2500.

The expected latency according to the change in $q$, the scale of $\boldsymbol{\mu}$, is depicted in Fig. \ref{Fig:Mu_vs_schemes}. If $q \le 10^{-2}$, then the uniform load allocation with $n^*$ appears to achieve the proposed lower bound. In the region $[10^{-1.5} , \; 10^{-1}]$, there is a tendency that the uniform load allocation with rate $\frac{1}{2}$ has a relatively lower expected latency than the other schemes, except for the proposed load allocation. On the other hand, in the regions excluding $[10^{-1.5} , \; 10^{-1}]$, it has a relatively high expected latency compared to the lower bound.
In addition, the result shows that the expected latency of an uncoded scheme using the rate 1 uniform load allocation method approaches the proposed lower bound $T^{\star}$ as $q$ increases to $10^{1.5}$. Based on these observations, we conducted the following experiments with the assumption that the rate of MDS code would have a significant impact on the expected latency.

%
%

Fig. \ref{Fig:Rate_proposed} shows the rate $\frac{k}{n^*}$ as a function of $q$, the scale of $\boldsymbol{\mu}$, with $\boldsymbol{N}$ as in the previous setting of Fig. \ref{Fig:Mu_vs_schemes}. In the region $[10^{-1.5} , \; 10^{-1}]$, the rate is close to $\frac{1}{2}$, and the rate is almost $0.99$ when $q = 10^{1.5}$. 
The following simulations are done in region $[10^{-2}, \; 10^{1.5}]$, since the $(n^*, k)$ MDS code achieves the proposed lower bound $T^{\star}$ if $q \le 10^{-2}$.
The simulation results under the uniform load allocation show that an $(n, k)$ MDS code can exist to have an expected latency lower than the expected latency of a system using the $(n^*, k)$ MDS code if we use an $(n, k)$ MDS code close to the optimal rate $\frac{k}{n^*}$.

Therefore, the effects of various rates on the expected latency of the uniform load allocation is depicted in Fig. \ref{Fig:Proposed_vs_uniform}.
The simulation result shows that when $q = 1$, the MDS code with rate $\frac{2}{3}$ has a lower expected latency than that of the optimal $(n^*, k)$ MDS code under the uniform load allocation. In addition, in Fig. \ref{Fig:Proposed_vs_uniform_900workers}, the simulation is performed under the condition that there are two groups with parameter $\boldsymbol{N} = (N_1, N_2) = (300,600 )$,  $\boldsymbol{\mu} = (\mu_{(1)}, \mu_{(2)}) = (4, 0.5)$, and $\boldsymbol{\alpha} = (\alpha_{(1)}, \alpha_{(2)}) = (1, 1)$. 
In the simulation, we have found that the lowest expected latency is achieved when the rate is near $0.52$ under the uniform load allocation.
Given the same parameters, the proposed load allocation shows a $10 \%$ reduction in the expected latency compared to the uniform load allocation with rate $0.52$.





We now proceed with the simulation for the probability model with a shift parameter. 
From (\ref{Eq:Reisizadeh}), we have $E[T_i^{a}] = l_{(j)} \left( \alpha_{(j)} + \frac{1}{\mu_{(j)}} \right)  .$
This means that $k$ is a scaling factor that affects the expected latency in the probability model with a shift parameter. 
The following simulations are conducted with $k$ fixed to $10^5$.
In Fig. \ref{Fig:Proposed_hetero3}, we set $\boldsymbol{N} = \frac{1}{10}(3N, 3N, 4N)$, $\boldsymbol{\mu} = (1, 4, 8) $, and  $\boldsymbol{\alpha} = (1, 4, 12)$. In this simulation, we observe that the proposed load allocation with a shift parameter achieves the lower bound $T_a^{\star}$. This result is consistent with the result of \cite{Reisizadeh19Coded} which is known to be an optimal load allocation scheme.

\begin{figure*}[t]
\centering
\begin{multicols}{2}
\includegraphics[width=1\linewidth]{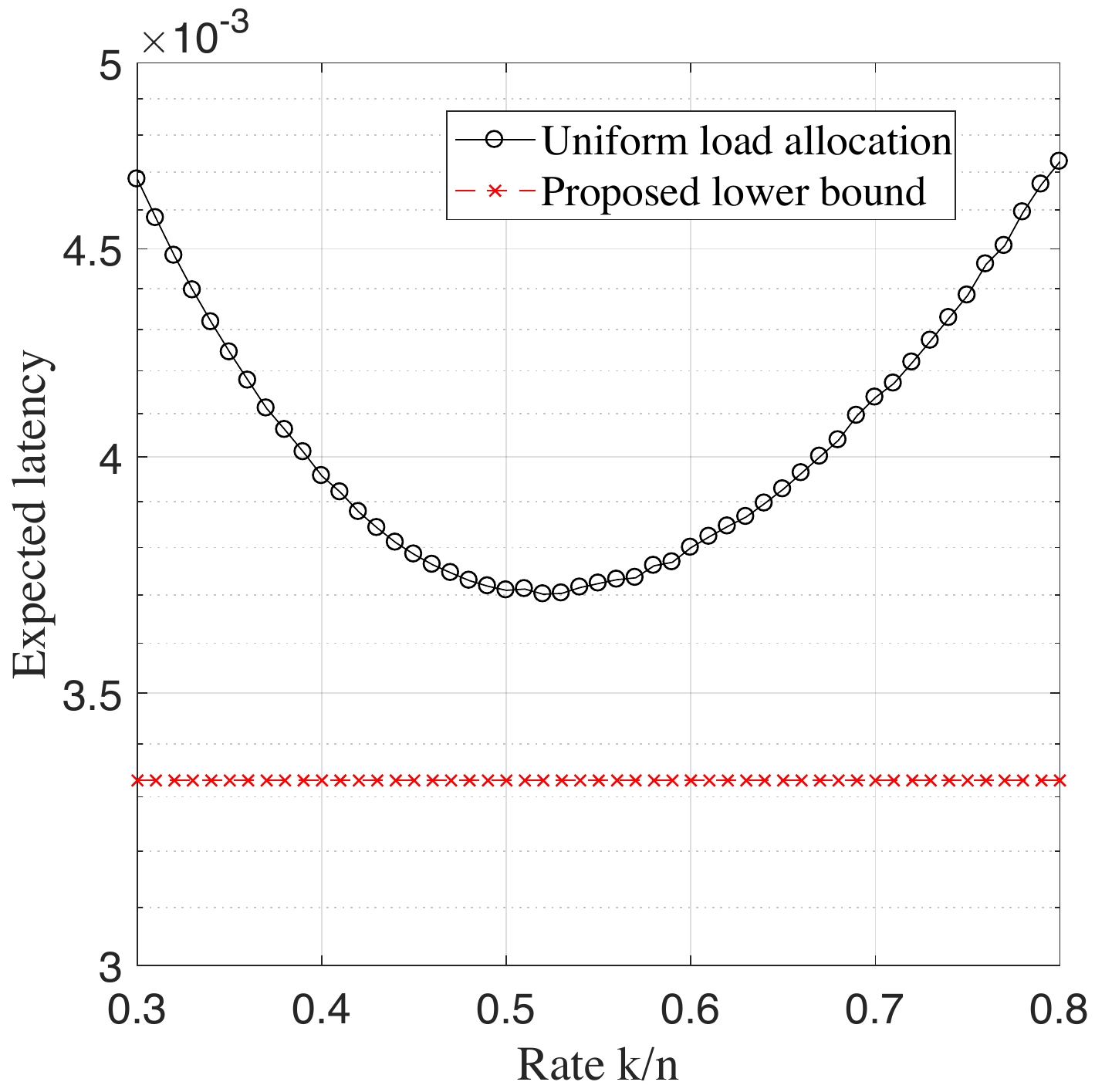} \par\caption{Expected latency according to change of rate with uniform load allocation.}
	\label{Fig:Proposed_vs_uniform_900workers}
\includegraphics[width=1\linewidth]{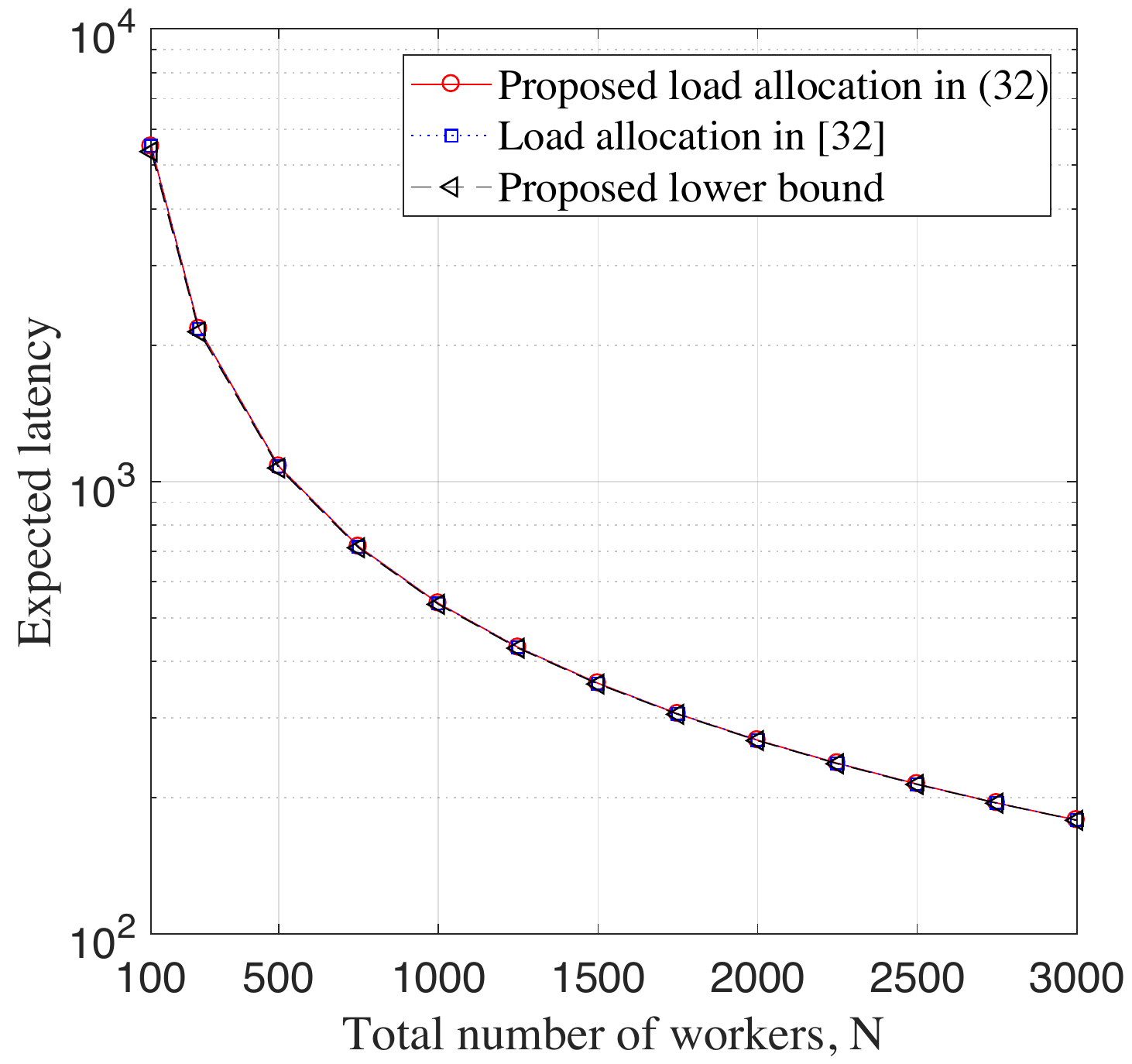} \par\caption{Comparison between the proposed load allocation in \eqref{Eq:Solutionl_j_a^*} and the load allocation algorithm in \cite{Reisizadeh19Coded} with three groups.}
	\label{Fig:Proposed_hetero3}
\end{multicols}
\vspace{-0.1in}
\end{figure*}

\section{Conclusion}
In this paper, we proposed the optimal load allocation
for distributed matrix-vector multiplication in heterogeneous group
clusters. Specifically, we established a lower bound on the expected latency and obtained the optimal load allocation
by showing that our proposed load allocation achieves the minimum of the lower bound, which is shown to be the theoretical limit.
Along with the proposed load allocation, the optimal design of the $(n, k)$ MDS code is obtained.
The optimal load allocation in our setting is a generalization of the result in \cite{Lee18Speeding}.
From numerical evaluations, it is shown that the proposed load allocation provides a 10x reduction in expected latency compared to the existing scheme. 



%

\appendices

\section{Derivation of equation \eqref{Eq:Group_expectation}} \label{App:Group_Expectation}

Let $f_{T_{r_j:N_j}^{l_{(j)}}} (t) $ denote the pdf of $T_{r_j:N_j}^{l_{(j)}}$, i.e., 
$$
f_{T_{r_j:N_j}^{l_{(j)}}} (t) = N_j  f_{j} \left ( t \right )  \binom{N_j -1}{r_j -1} F_{j} \left (t \right )^{r_j -1} \left [1 - F_{j} \left ( t \right ) \right ]^{N_j - r_j} , 
$$
where $f_j(t) = F_j'(t)$ and $t \ge \frac{\alpha_{(j)} l_{(j)}}{k}$.
Then 
\begin{align} \label{Eq:APP1}
\lambda_{r_j:N_j}^{l_{(j)}} & =  \int_{\tau = \frac{\alpha_{(j)} l_{(j)}}{k}}^{\infty} \tau f_{T_{r_j:N_j}^{l_{(j)}}} (\tau) d\tau .
\end{align}
Let $\frac{k}{l_{(j)}} \tau - \alpha_{(j)}  = x$.  This gives that $\frac{k}{l_{(j)}}d\tau = dx$.
Then we have
\begin{align*}
\lambda_{r_j:N_j}^{l_{(j)}} & =  \int_{x = 0}^{\infty} \frac{l_{(j)}}{k} (\alpha_{(j)}+x) f_{T_{r_j:N_j}^{l_{(j)}}} \left (\frac{l_{(j)}}{k} (x + \alpha_{(j)}) \right ) \frac{l_{(j)}}{k} dx.
\end{align*}
We easily check that 
$$
f_j \left (\frac{l_{(j)}}{k} (x + \alpha_{(j)}) \right ) = \frac{k\mu_{(j)}}{l_{(j)}} e^{-\mu_{(j)} x} \hspace{0.2in}\textnormal{ and } \hspace{0.2in} F_j \left (\frac{l_{(j)}}{k} (x + \alpha_{(j)}) \right ) = 1-e^{-\mu_{(j)} x}.
$$

Let $f(x) = \mu_{(j)} e^{-\mu_{(j)} x}$ and $F(x) = 1-e^{-\mu_{(j)} x}$.
Then \eqref{Eq:APP1} is rewritten as 
\begin{align*}
 \lambda_{r_j:N_j}^{l_{(j)}}  =  \int_{x = 0}^{\infty} \frac{l_{(j)}}{k} (\alpha_{(j)} + x) N_j f \left ( x \right )  \binom{N_j -1}{r_j -1} F \left ( x \right )^{r_j -1} \left [1 - F(x) \right ]^{N_j - r_j} dx .
\end{align*}
Therefore, we obtain 
$$
 \lambda_{r_j:N_j}^{l_{(j)}} = \frac{l_{(j)}}{k} \left (\alpha_{(j)} + \frac{\mathcal{H}(N_j ) - \mathcal{H}(N_j - r_j)}{\mu_{(j)}} \right )
$$
as desired.

%
%
%
%
%
%

\section{Proof of Theorem \ref{Thm:Lower_condition}} \label{APP:Theorem1}
\begin{proof}
Without loss of generality, assume that $\lambda_{r_1^*: N_1}^{l_{(1)}^*} \ge \lambda_{r_2^*: N_2}^{l_{(2)}^*} \ge \dots \ge \lambda_{r_G^*: N_G}^{l_{(G)}^*}  .
$
Let $E$ be the cardinality of the set $e$, denoted by $|e|$, where $$e = \left\{ j'\in [G] : \lambda_{r_{j'}^* : N_{j'} }^{l_{(j')}^*} = \max_{j\in [G]} \left\{ \lambda_{r_j^*: N_j}^{l_{(j)}^*} \right\}  \right\}  .$$
Assume that $E < G$. Then there exists $h \in [G]\backslash [E]$ such that 
\begin{equation*}
\epsilon  := \lambda_{r_{j'}^* : N_{j'} }^{l_{(j')}^*} - \lambda_{r_{h}^* : N_{h} }^{l_{(h)}^*} > 0, \hspace{0.5in} \textnormal{ for } j'\in [E]  .
\end{equation*}
Note that the pair $(\boldsymbol{l}^*, \boldsymbol{r}^*)$ satisfies the constraint $\sum_{j\in [G]} r_j^* l_{(j)}^*  = k$.
For fixed $\boldsymbol{r}^*$, the following procedure is repeated $E$ times to obtain a pair $\left( \bar{l}_{(s)}, \bar{l}_{(h)}^{(s)} \right)$ in ascending order with respect to $s$, for $s \in [E]$, such that $$\lambda_{r_h^*: N_h}^{\bar{l}_{(h)}^{(s)}} = \lambda_{r_h^*: N_h}^{\bar{l}_{(h)}^{(s - 1)}} + \frac{\epsilon}{2E}, $$
where
\begin{equation*}
\bar{l}_{(j)}^{(0)} = l_{(j)}^* \hspace{0.2in} \textnormal{ and } \hspace{0.2in} \bar{l}_{(s)} = \frac{k - \sum_{j\in [s-1]}r_j^* \bar{l}_{(j)} - \sum_{j\in [G]\backslash ([t] \cup \{h\} )}  r_j^* l_{(j)}^* - r_h^* \bar{l}_{(h)}^{(s)} }{r_s^*} .
\end{equation*}

Then get a vector $\hat{\boldsymbol{l}} = (\hat{l}_{(1)}, \hat{l}_{(2)}, \dots, \hat{l}_{(G)})$, where
\begin{eqnarray*}
\hat{l}_{(j)} = \left \{
                 \begin{array}{llll}
                 \displaystyle 
             \bar{l}_{(j)}, & & \hbox{\textnormal{if } $ j \in [E]$ ,}& \\
                   \bar{l}_{(h)}^{(E)}, & & \hbox{\textnormal{if } $j = h$ ,}& \\
                    l_{(j)}^*, & & \hbox{\textnormal{if } $j \in [G]\backslash ([E]\cup \{h \})$ .}& 
                  \end{array}
                \right.
\end{eqnarray*}
The existence of $\hat{\boldsymbol{l}}$ is the contradiction to the the assumption that $(\boldsymbol{l}^*, \boldsymbol{r}^*)$ achieves the minimum of $\max_{j\in [G]} \left\{ \lambda_{r_j: N_j}^{l_{(j)}} \right\}$.
\end{proof}

\section{Proof of the concavity of $g_j(\boldsymbol{r})$ in Lemma \ref{Lmm:Convexity_f(r)}} \label{App:g_j}
\begin{proof}
It suffices to show that 
\begin{equation*}
h(x) = \frac{x}{\alpha + \frac{1}{\mu} \log \left(\frac{N}{N - x} \right)}
\end{equation*}
is a concave function with respect to $x$ on the interval $[0, N)$.
The second derivative of $h(x)$ is 
\begin{equation} \label{Eq:Sec_derivative_h}
-\frac{\mu \left( \alpha \mu(2 N- x) +(2N-x) \log \left( \frac{N}{N-x} \right) -2x \right)}{(N-x)^2 \left( \alpha \mu + \log \left( \frac{N}{N-x} \right) \right)^3}  .
\end{equation}
Since $ \alpha \mu(2 N - x)$ and the denominator in (\ref{Eq:Sec_derivative_h}) are positive, it suffices to show that 
\begin{equation*}
v(x) := (2N-x) \log \left( \frac{N}{N-x} \right) -2x \ge 0  .
\end{equation*}
We easily check that $v'(x) =  \frac{x}{N-x} - \log \left( \frac{N}{N-x} \right) \ge 0 $
on  $[0, N)$ and $v(0) = 0$.
\end{proof}

\section{Load allocation method proposed in \cite{Reisizadeh19Coded}} \label{App:Coded}

Let  $\tilde{l}_{(j)}$ denote the load allocation proposed in \cite{Reisizadeh19Coded}. Then  $\tilde{l}_{(j)} = \frac{k}{s \delta_{(j)} },$
where
\begin{equation*}
 s = \sum_{i \in [N]} \frac{\mu_i}{1 + \mu_i \delta_i} \hspace{0.2in} \textnormal{ and } \hspace{0.2in} \delta_{(j)} = - \frac{W_{-1} (- e^{-( \alpha_{(j)} \mu_{(j)} + 1)}) + 1 }{ \mu_{(j)}} \hspace{0.2in} \textnormal{ for } j\in [G] .
\end{equation*}
Depending on our system model, $s$ can be rewritten as $$s = \sum_{j\in [G]} N_j \left( \frac{\mu_{(j)}}{1 + \mu_{(j)}\delta_{(j)}} \right)  . $$
Thus, we have $$\tilde{l}_{(j)} = \frac{k}{ \delta_{(j)} \sum_{j'\in [G]} N_{j'} \left( \frac{\mu_{(j')}}{1 + \mu_{(j')}\delta_{(j')}} \right)} .$$
We use an $(\tilde{n},k)$ MDS codes for this load allocation, where $\tilde{n} = \sum_{j\in [G]} N_j \tilde{l}_{(j)} .$


%
%

\ifCLASSOPTIONcaptionsoff
  \newpage
\fi

\bibliographystyle{IEEEtranTCOM}
\bibliography{IEEEabrv,loadAlloc19}

\end{document}